\documentclass{article}
\usepackage{graphicx} 
\usepackage{amsmath}
\usepackage{amssymb}
\usepackage{graphicx}
\usepackage{float}
\usepackage{caption}
\usepackage{amsthm}
\usepackage{color}
\usepackage{tabu}
\usepackage{hyperref}
\usepackage{cleveref}
\usepackage{thmtools}
\usepackage{thm-restate}
\usepackage{fullpage}
\usepackage{mathtools}
\usepackage{microtype}

\newcommand{\rr}{\mathbb R}

\theoremstyle{plain}
\newtheorem{theorem}{Theorem}
\newtheorem{lemma}[theorem]{Lemma}

\theoremstyle{definition}

\title{NP-Membership for the Boundary-Boundary Art-Gallery Problem}
\author{Jack Stade}
\date{November 2025}

\begin{document}

\maketitle

\begin{abstract}
The boundary-boundary art-gallery problem asks, given a polygon $P$ representing an art-gallery, for a minimal set of guards that can see the entire boundary of $P$ (the wall of the art gallery), where the guards must be placed on the boundary. That is, for each point on the boundary, there should be a line segment connecting it to one of the guards that is contained in $P$. We show that this art-gallery variant is in NP, even if the polygon can have holes. In order to prove this, we develop a constraint-propagation procedure for continuous constraint satisfaction problems where each constraint involves at most 2 variables.

The X-Y variant of the art-gallery problem is the one where the guards must lie in X and need to see all of Y. Each of X and Y can be either the vertices of the polygon, the boundary of the polygon, or the entire polygon, giving 9 different variants. Previously, it was known that X-vertex and vertex-Y variants are all NP-complete and that the point-point, point-boundary, and boundary-point variants are $\exists \mathbb{R}$-complete [Abrahamsen, Adamaszek, and Miltzow, JACM 2021][Stade, SoCG 2025]. However, the boundary-boundary variant was only known to lie somewhere between NP and $\exists \mathbb{R}$. 

The X-vertex and vertex-Y variants can be straightforwardly reduced to discrete set-cover instances. In contrast, we give example to show that a solution to an instance of the boundary-boundary art-gallery problem sometimes requires placing guards at irrational coordinates, so it unlikely that the problem can be easily discretized. 
\end{abstract}

\tableofcontents

\newpage
\clearpage
\pagenumbering{arabic}

\setcounter{page}{1}

\section{Introduction}

\paragraph{The art-gallery problem and variants.}

Given a plane polygon $P$, two points $x$ and $y$ in $P$ are said to be visible to each other if the line segment between $x$ and $y$ is contained $P$. The art-gallery problem \cite{ArtGalleryTextbook} is a decision problem that asks, given a polygon $P$ (called the art-gallery), whether there are $n$ points in $P$ (called the guards) such that every point in $P$ is visible to at least one of the guards.

For $X, Y\in \{\text{Vertex, Boundary, Point}\}$, the X-Y art-gallery problem \cite{XYDefinition} is a variant of the art-gallery problem where the guards are restricted to lie in X, but only need to see Y. The original variant (as described above) is called the point-point variant under this system. The variant where the guards only need to see the boundary of the polygon is called the point-boundary variant.

All of these variants are at least NP-hard \cite{NPHardness,BoundaryNPHard}, but it isn't clear that they are necessarily contained in NP. The coordinates of the guards can be arbitrary real numbers, and some instances require guards to be placed at irrational coordinates in order to be guarded optimally, even when the vertices of $P$ are given as integers (see \cite{IrrationalGuards}).

\paragraph{The complexity class $\exists\mathbb{R}$.}

The problem ETR (short for Existential Theory of the Reals) asks whether a statement of form:

\[\exists x_1, \dots, x_n\in \rr:\Phi(x_1, \dots, x_n)\]

\noindent is true, where $\Phi$ is a boolean formula involving the signs of integer-coefficient polynomials in the variables $x_1, \dots, x_n$. The complexity class $\exists \mathbb{R}$ consists of problems that can be polynomial-time reduced to ETR. It is straightforward to see that $\text{NP}\subseteq \exists\mathbb{R}$. In 1988, Canny \cite{ETRPSPACE} established that $\exists \mathbb{R}\subseteq \text{PSPACE}$. It is unknown whether either inclusion is strict. 

\paragraph{Classification of art-gallery variants.}

In 2018, Abrahamsen, Adamaszek and Miltzow \cite{ExistsRHardness} showed that the point-point and boundary-point art-gallery variants are $\exists\rr$-complete, meaning in particular that they are not contained in NP unless NP$=\exists\rr$. Stade \cite{Stade2025} later proved the point-boundary variant is also $\exists\rr$-complete. In contrast, all the vertex-Y and X-vertex variants are easily seen to be in NP. We show that the boundary-boundary variant is in NP, completing the classification of the complexities of X-Y art-gallery variants (at least up to $\text{NP}\stackrel{?}{=}\exists \rr$). This is summarized in \Cref{tab:InitialsTable}.

\begin{table}[ht]
\begin{center}
\begin{tabu}{|c|c|}
\hline
Variant&Complexity\\
\hline
Vertex-Y&NP\cite{NPHardness}\\
X-Vertex&NP\cite{NPHardness}\\
Point-Point&$\exists\rr$\cite{ExistsRHardness}\\
Boundary-Point&$\exists\rr$\cite{ExistsRHardness}\\
Point-Boundary&$\mathbf{\exists\rr}$\cite{Stade2025}\\
Boundary-Boundary&NP[{\bf this paper}]\\
\hline
\end{tabu}
\caption{Complexity of X-Y art-gallery variants}
\label{tab:InitialsTable}
\end{center}
\end{table}

\subsection{Background and related work}

\paragraph{Art galleries with contiguous guarding.}

Biniaz, Maheshwari, Merrild, Mitchell, Odak, Polishchuk, Robson, Rysgaard, Schou, Shermer, Spalding-Jamieson, Svenning, and Zheng \cite{ContiguousArtGallery} recently gave a polynomial-time algorithm for the \emph{contiguous art-gallery problem}, a variant of the point-boundary art-gallery where each guard can be responsible only for a connected section of the polygon boundary. 

\paragraph{$\exists\mathbb{R}$-completeness.}

In addition to art galleries, many other important problems are known to be $\exists\rr$-complete, notable examples including packing \cite{PackingETRHard}, drawing graphs with fixed edge lengths \cite{Schaefer2013}, realizeability of line arrangements \cite{Mnev1988,Shor1990}, and training neural networks \cite{NNERHard,ConnectedNNERHard}. A recent survey by Schaefer, Cardinal, and Miltzow and \cite{ExistsRSurvery} lists roughly 150 $\exists\rr$-complete problems.

\paragraph{Constraint satisfaction problems.}

A constraint satisfaction problem (CSP) is a decision problem that asks whether there is an assignment of some variables that satisfies a set of constraints. A 2CSP is a CSP where each constraint involves at most 2 variables. 

A \emph{continuous constraint satisfaction problem} (CCSP) is a constraint satisfaction problem where the variables take values over a continuous domain, typically $\mathbb{R}$ or a subset thereof. The constraints are typically represented by equalities and inequalities involving polynomials in the variables. A 2CCSP is a CCSP where each constraint involves at most 2 variables.

Discrete CSPs are typically in NP and are frequently NP-complete. Similarly, CCSPs are typically in $\exists \mathbb{R}$ and are frequently $\exists\mathbb{R}$-complete. Miltzow and Schmiermann \cite{CCSPClassification} showed that CCSPs typically become $\exists\rr$-hard if constraints of form $x+y=z$ can occur along with some nonlinear constraints. They conjecture that it is necessary to have a constraints that depends on $3$ or more variables in order for a CCSP to be $\exists\rr$-hard, as current techniques seem to be unable to show $\exists\rr$-hardness for any 2CCSPs.

\paragraph{2CSPs.}

It has been known since at least 1967 \cite{Binary2SAT} that boolean (variables take values in $\{0, 1\}$, i.e. 2SAT) 2CSPs can be decided in polynomial time. On the other hand, ternary (variables take values in $\{0, 1, 2\}$) 2CSPs are already NP-hard \cite{Monotone2SAT}. In 2000, Beckert, Hahnle and Manya \cite{Monotone2SAT} gave a polynomial-time algorithm for solving (discrete) 2CSPS where the variables lie in a lattice (that is, a partially ordered set where least upper bounds and greatest lower bounds exist) and the constraints are disjunctions of terms of form $x\ge c$ or $x\le c$ (the constraints are in some sense \emph{monotone} with respect to the lattice). Charatonik and Wrona \cite{MonotoneFaster} later gave a faster algorithm for these instances. 

\paragraph{Linear 2CSPs.}

Linear programming is essentially a type of CCSP. In 1983, Megiddo \cite{LinearFeasibility} gave a strongly polynomial time algorithm for determining feasibility of a linear programming instance where each constraint involves at most $2$ variables. In 2023, Dadush, Koh, Natura, Olver, and V\'{e}gh \cite{LinearOptimality} gave a strongly polynomial time algorithm for determining optimality for these instances. The best known algorithms for general linear programming run in only weakly polynomial time, in the sense that the number of arithmetic operations needed depends on the bit-complexity of the input coefficients. 

\paragraph{Constraint propagation.}

A common technique for solving CSPs is \emph{constraint propagation}: given an instance $\Gamma$, we can derive new constraints from the existing ones. A method for producing new constraints giving existing ones is called an \emph{inference rule}. If we can derive the unsatisfiable clause $\emptyset$ from the original instance using some set of inference rules, then we conclude that the instance is unsatisfiable. A derivation of $\emptyset$ is called a \emph{refutation} of the instance. A set of inference rules for a CSP is \emph{refutation complete} if every unsatisfiable instance has a refutation.

In 1960, Davis and Putman \cite{ResolutionCompleteness} proved refutation-completeness for the following inference rule on SAT formulas:

\[\frac{x\vee C\quad \neg x\vee D}{C\vee D}\]

The notation here means that we can derive the constraint below the bar from the constraints above the bar. This particular inference rule is called binary resolution. It produces a constraint $C\vee D$ given constraints $x\vee C$ and $\neg x\vee D$, where $x$ is a variable and $C$ and $D$ are (possibly empty) constraints. In general, there is no guarantee that a refutation has polynomial size or that one can be found efficiently if it exists. However, resolution leads to a polynomial-time algorithm for boolean 2SAT (see Krom 1967 \cite{Binary2SAT}).

In 2000, Beckert, Hahnle and Manya \cite{Monotone2SAT} extended binary resolution to the case where the variables take values in a finite lattice $N$, and the constraints are disjunctions of terms of form $x\le c$ or $x\ge c$, where $x$ is a variable and $c\in N$. Specifically, they show that the inference rule:

\[\frac{(x\le c)\vee C\quad (x\ge d)\vee D}{C\vee D}\quad \text{where }d>c\]

\noindent is refutation complete for these CSPs. In the case where each constraint involves at most 2 variables, they show that refutations can be found in polynomial-time (if they exist).

\paragraph{Quantifier elimination.}

A family of techniques for solving CCSPs is quantifier elimination: given some subset $S\subset \{x_1, \dots, x_n\}$, we can try to find conditions that determine whether a given assignment of the variables in $S$ extends to a satisfying assignment of the instance. The variables in $\{x_1, \dots, x_n\}\setminus S$ are said to be eliminated. If we can eliminate all the variables, then we can decide the truth value of the instance. 

In 1948, Tarski \cite{QuantifierElimination} gave the first quantifier-elimination procedure for ETR. Given a formula:

\[\exists x_n:\Phi(x_1, \dots, x_{n})\]

\noindent he shows how to compute an equivalent quantifier-free formula $\Psi(x_1, \dots, x_{n-1})$. The formula $\Psi$ represents a sort of algorithm for determining if an assignment of the variables $x_1, \dots, x_{n-1}$ can be extended to an assignment of $x_1, \dots, x_n$ that satisfies $\Phi$. The procedure can be repeated to eliminate any subset of the variables, though each elimination causes an exponential blow-up in the size of the formula. 

Even if computing it is infeasible, the existence of a quantifier elimination can sometimes be used to bound the complexity of solutions if they exist. For example, if we eliminate all but one variable $x$ using Tarski's quantifier elimination procedure for ETR, then we obtain a formula only involving polynomials in $x$, showing that if a satisfying assignment exists, then one exists where $x$ is an algebraic number. More usefully, Schaefer and Stefankovič \cite{ETRBounding} use the existence of a certain quantifier elimination to show that ETR is $\exists\rr$-complete even when the variables are restricted to the range $[-1, 1]$.

\subsection{Our results}

We show that the boundary-boundary art-gallery problem is in NP, even for polygons with holes. The proof is by a nondeterministic reduction to a particular 2CCSP problem that we call $\mathcal{M}-2\text{SAT}$, which is a 2CCSP where the constraints are given by (piecewise) fractional-linear functions. That is to say, the constraints are piecewise of form:

\[x\le \frac{ay+b}{cy+d}\]

\noindent where $x$ and $y$ are variables and $a, b, c$ and $d$ are integer constants.

\begin{restatable}{theorem}{msatreduction}\label{thm:BBreduction}
The boundary-boundary art-gallery art-gallery problem reduces in NP to $\mathcal{M}-2\text{SAT}$, even for polygons with holes
\end{restatable}

A precise definition of $\mathcal{M}-2\text{SAT}$ and a proof of \Cref{thm:BBreduction} are given in \Cref{sec:m2sat}. Most of the paper is devoted to showing that $\mathcal{M}-2\text{SAT}$ is in NP. Eventually, we will show that, if a solution exists, then there is a solution where the variables have form $p+q\sqrt{r}$, where $p$, $q$ and $r$ are rational numbers with a polynomially-bounded number of bits. Such solutions can be verified efficiently. 

In order to do this, in \Cref{sec:ctsinference}, we define a a set of inference rules for monotone 2-CCSPs. We say that a \emph{continuous formula} on variables $v_1, \dots, v_n$ is a conjunction of constraints of form $x\le f(y)$, where $f$ is an increasing bijection $\rr\rightarrow \rr$ and $x, y\in \{\pm v_1, \dots, \pm v_n\}$, and for each each $v_i$ a constraint $v_i\in \text{range}(v_i)$, where $\text{range}(v_i)$ is a compact interval.

We define three inference rules for continuous formulas. The first inference rule is function composition, which can be expressed as

\[\frac{x\le f(y) \quad y\le g(z)}{x\le f(g(z))}\]

This means that if, if we have constraints $x\le f(y)$ and $y\le g(z)$, then we can generate a new constraint $x\le f(g(z))$. 

The remaining two inference rules detect cases where a formula is unsatisfiable. The first of these can be written

\[\frac{x\le f(y)}{\emptyset}\text{ where }f(\text{max}(\text{range}(y)))<\text{min}(\text{range}(x))\]

This means that if we have a constraint $x\le f(y)$ where the largest possible value of $y$ results in $f(y)$ being less than the smallest possible value of $x$, then we can conclude that the formula is unsatisfiable. 

The final inference rule can be written

\[\frac{x\le f(-x)\quad -x\le g(x)}{\emptyset}\text{ where }\exists c: f(-c)<c\text{ and }g(c)<-c\]

This third rule is required to detect the relationship between $x$ and $-x$. If we have constraints $x\le f(-x)$ and $-x\le g(x)$ where $f(-c)<c$ and $g(c)>-c$ for some $c$, then the formula is unsatisfiable. 

The key result of \Cref{sec:ctsinference} is that these three inference rules are enough to determine the satisfiability of certain continuous formulae. Note that a constraint $x\le f(y)$ is equivalent to $f^{-1}(x)\le y$ (since $f$ is an increasing bijection). We say that a continuous formula is symmetric if, for every constraint $x\le f(y)$ that appears, the equivalent $f^{-1}(x)\le y$ also appears. 

\begin{restatable}{theorem}{resolutioncompleteness}\label{lem:ResolutionCompleteness}
Let $\Gamma$ be an unsatisfiable symmetric continuous formula. Then the empty constraint can be derived from $\Gamma$ by the continuous inference rules.
\end{restatable}

The proof of \Cref{lem:ResolutionCompleteness} uses a compactness argument to turn a continuous formula into a discrete formula. A result of Beckert, Hahnle and Manya \cite{Monotone2SAT} shows that binary resolution is refutation-complete for the discrete instance, and we show that we can simulate binary resolution using the continuous inference rules.

The nature of a compactness argument means that we don't get any useful bound on the length of refutations needed by our system. So in \Cref{sec:logdepth} we extend our inference system to allow infinite compositions and $n$-way minimum operations. In the new inference system, we prove that it is sufficient to consider refutations that in some sense have logarithmic depth. 

In \Cref{sec:algorithms}, we use these results to give two upper bounds for $\mathcal{M}-2\text{SAT}$. First, we show that $\mathcal{M}-2\text{SAT}\in \text{NP}$. Second, we show that there is a quasi-polynomial-time algorithm for $\mathcal{M}-2\text{SAT}$. 

The problem $\mathcal{M}-2\text{SAT}$ will contain a range constraint $x\in [a, b]$ for each variable. The construction in \Cref{sec:logdepth} is used to construct a sort of circuit that determines the satisfiability of the instance given the values of $a$ and $b$. If we set $a=b=c$, then this circuit determines whether $x=c$ extends to a satisfying assignment of the instance, giving a method of quantifier elimination. We can bound the complexity of arithmetic operations performed by the circuit, giving us a bound on the complexity of solutions needed. We obtain $\mathcal{M}-2\text{SAT}\in \text{NP}$. 

The quasi-polynomial-time algorithm for $\mathcal{M}-2\text{SAT}$ is essentially straightforward to obtain from the results of \Cref{sec:logdepth}. 

It it is easy to construct instances of $\mathcal{M}-2\text{SAT}$ where a solution requires irrational coordinates, but it is not immediately clear if such an instance can be realized as an instance of the boundary-boundary art-gallery problem. In \Cref{sec:irrationalexample}, we give an example to show that an optimal solution to the boundary-boundary art-gallery problem can require irrational coordinates.

\section{The boundary-boundary art-gallery problem and $\mathcal{M}-2\text{SAT}$}\label{sec:m2sat}

Let $\mathcal{F}$ be a class of continuous increasing bijections $\rr\rightarrow \rr$. Then an instance of $\mathcal{F}-2\text{SAT}$ consists of a set of variables $v_1, \dots, v_n$, a closed interval $\text{range}(v_i)$ for each variable, and a set of constraints of form $x\le f(y)$ where $x, y$ are \emph{literals}, that is elements of $\{\pm v_1, \dots, \pm v_n\}$. The problem asks whether there is an assignment of real values $v_i\in \text{range}(v_i)$ that satisfy all the constraints. We write $\text{range}(-v_i)=-\text{range}(v_i)$.

\Cref{fig:ConstraintExamples} shows some examples of the types of sets that can occur. Even though the variable $v_i$ will only ever take values in $\text{range}(v_i)$, it is helpful to have the functions defined on all of $\rr$ so that inverses and function compositions are well-defined.

\begin{figure}
\centering
\includegraphics[page=1]{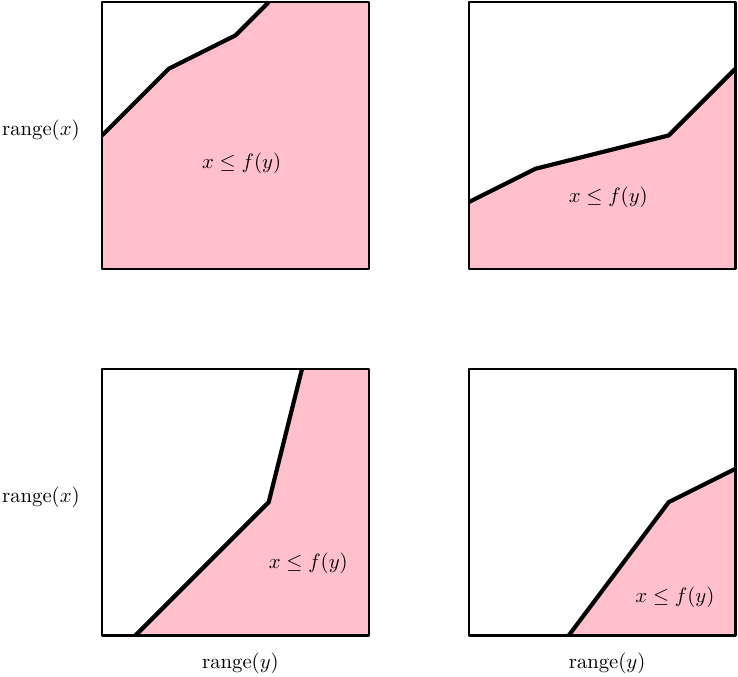}
\caption{Some examples of constraints of the form $x\le f(y)$ for different $f$. }
\label{fig:ConstraintExamples}
\end{figure}

In many cases, the constraint $x\le f(y)$ already restricts $x$ or $y$ to a smaller interval. We could try to restrict the variables to this range as a pre-processing step. However, restricting the range of one variable might then require restricting the range of another variable, and this can lead to an infinite loop. This process does converge in some sense, but it isn't clear that we can compute the limit.

If the functions in $\mathcal{F}$ are very complicated, for example polynomial circuits, then there is a little hope of being able to decide $\mathcal{F}-2\text{SAT}$ efficiently. The specific class of functions that we are most interested in are \emph{increasing continuous piecewise fractional linear}, that is (increasing, continuous) functions that are piecewise of form:

\[f(x)=\frac{ax+b}{cx+d}\]

\noindent where $a, b, c$ and $d$ are integers. We let $\mathcal{M}$ be this class of functions, so $\mathcal{M}-2\text{SAT}$ is the 2CCSP using functions from $\mathcal{M}$ to define the constraints. 

\msatreduction*


\begin{proof}
If $P$ is an art-gallery with $n$ vertices, then the boundary of $P$ can be easily guarded with $n$ (or even $\left\lceil\frac{n}2\right\rceil$) boundary guards. In a guarding configuration of $P$, the portion of the boundary that can be seen by any given guard is composed of at most $n$ connected regions. This is because any two different connected components of the boundary region must be separated by a vertex of $P$, as illustrated in \Cref{fig:BBreduction} (left). So in a guarding configuration, each edge of $P$ is covered by a set of at most $n^2$ intervals that are each guarded by a single guard. 

\begin{figure}
\centering
\includegraphics[page=1]{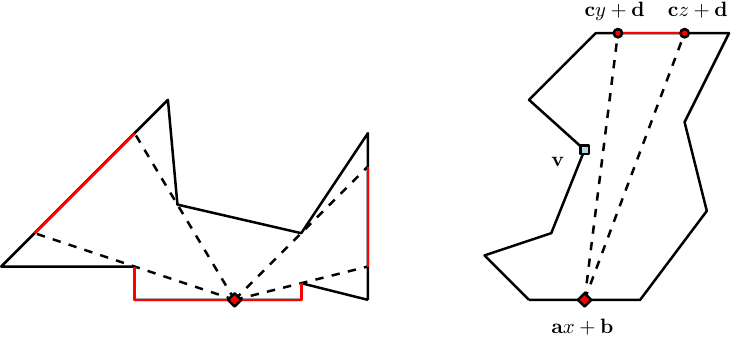}
\caption{Left: the portion of the boundary seen by a single guard is a union of at most $n$ connected pieces, since any two different pieces are separated by a vertex of the polygon. There are at most $n$ guards, so in a guarding configuration, each edge is covered by a union of at most $n^2$ intervals that are each visible to a single guard. Right: we guess that a guard at position $\mathbf{a}x+\mathbf{b}$ sees the interval from $\mathbf{c}y+\mathbf{d}$ to $\mathbf{c}z+\mathbf{d}$. In order for this to be true, the vertex $\mathbf{v}$ of $P$ must be on the left side of the $\mathbf{a}x+\mathbf{b}$ to $\mathbf{c}y+\mathbf{d}$, creating a fractional-linear constraint.}
\label{fig:BBreduction}
\end{figure}

First, we nondeterminstically guess the structure of a guarding configuration. That is, for each guard, we guess which edge of $P$ it is on, and for each edge of $P$ we guess a partition of it into at most $n^2$ intervals (called visibility intervals), and for each visibility interval we guess which guard should guard it. The exact locations of each guard and the endpoints of the guarding intervals are not guessed, but will become variables in an instance of $\mathcal{M}-2\text{SAT}$.

For each edge of $P$, choose vectors $\mathbf{a},\mathbf{b}\in \mathbb{R}^2$ such that points on that edge have form $\mathbf{a}x+\mathbf{b}$ for $x\in [0, 1]$. We use this to parameterize the positions of the guards and the endpoints of the edges.

For each visibility interval, there is a constraint that the interval is visible to the chosen guard. We need to show that these constraints can be written using the functions from $\mathcal{M}$. 

To check that an interval from a point $\mathbf{y}$ to a point $\mathbf{z}$ is visible to a guard at position $\mathbf{x}$, we need to check that the triangle $\Delta\mathbf{xyz}$ is inside the polygon. To do this, if is sufficient to check that

\begin{itemize}
    \item each segment of the polygon boundary is interior-disjoint from $\Delta\mathbf{xyz}$, and
    \item the points $\mathbf{y}$ to a point $\mathbf{z}$ are on the inside-side of the segment containing $\mathbf{x}$
\end{itemize}

The triangle $\Delta\mathbf{xyz}$ is convex and each line segment on the boundary is convex, so if they are separated then there is a line separating them. The line containing the boundary segment or one of the lines containing a side of $\Delta\mathbf{xyz}$ is sufficient. For each segment, we nondeterministically guess which of these $4$ lines to use.

So there are two types of conditions that we might need to check:

\begin{itemize}
    \item check that a guard point or visibility interval endpoint is on the inside-side or outside-side of the line containing a polygon edge, or
    \item check that a polygon vertex is on the outside-side of a line from a guard vertex to an interval endpoint
\end{itemize}

The first type of constraint only depends on the position of one guard vertex or interval endpoint, say $\mathbf{a}x+\mathbf{b}$, and can be written $x\le c$ or $x\ge c$ for some rational constant $c$. These constraints can be implemented by adjusting the range of the variable $x$. 

The second type of constraint is more complicated, and is illustrated in \Cref{fig:BBreduction} (right). In order for the guard at position $\mathbf{u}=\mathbf{a}x+\mathbf{b}$ to see the line segment from $\mathbf{w}=\mathbf{c}y+\mathbf{d}$ to $\mathbf{c}z+\mathbf{d}$, we need to know that the segments from $\mathbf{a}x+\mathbf{b}$ to $\mathbf{c}y+\mathbf{d}$ or $\mathbf{c}z+\mathbf{d}$ are contained in the polygon $P$. 

The vertex marked $\mathbf{v}$ in the figure should be on the left side of the line from $\mathbf{a}x+\mathbf{b}$ to $\mathbf{c}y+\mathbf{d}$. This condition can be written

\begin{equation}\label{eq:visibilitycondition}
(\mathbf{a}x+\mathbf{b}-\mathbf{v})\times (\mathbf{c}y+\mathbf{d}-\mathbf{v}) \ge 0
\end{equation}

Here $\times$ represents cross product. The cross product of two vectors $\mathbf{p}$ and $\mathbf{q}$ in $\mathbb{R}^2$ is a scalar value given by the determinant of the matrix with columns $\mathbf{p}$ and $\mathbf{q}$. The condition \eqref{eq:visibilitycondition} is equivalent to

\[\mathbf{a}\times (\mathbf{c}y+\mathbf{d}-\mathbf{v})x \ge -(\mathbf{b}-\mathbf{v})\times (\mathbf{c}y+\mathbf{d}-\mathbf{v})\]

Let $a=-(\mathbf{b}-\mathbf{v})\times \mathbf{c}, b=-(\mathbf{b}-\mathbf{v})\times (\mathbf{d}-\mathbf{v}), c=\mathbf{a}\times \mathbf{c}$ and $d=\mathbf{a}\times (\mathbf{d}-\mathbf{v})$. Depending on the sign of $\mathbf{a}\times (\mathbf{c}y+\mathbf{d}-\mathbf{v})$ (that is, $cy+d$), we want to create the constraint either $x\ge \frac{ay+b}{cy+d}$ or $x\le \frac{ay+b}{cy+d}$. In the example in \Cref{fig:BBreduction} (right), $cy+d$ is always positive, but in general it could have different signs for different values of $y$.

We define a function $f$ that is equal to $\frac{ay+b}{cy+d}$ when $cy+d> 0$ and $0\le ay+b\le cy+d$ (so $\frac{ay+b}{cy+d}\in [0, 1]$ when $y$ satisfies these conditions). If there is a point $cy+d=0$ and $ay+b=0$, then $ay+b$ is divisible (as a polynomial in $y$) by $cy+d$, so $f$ is constant, and we can extend $f$ so that it is defined at the ``bad'' point. 

The conditions $cy+d\ge 0$ and $0\le ay+b\le cy+d$ are convex, so the set of values of $y$ satisfying them is an interval if it is non-empty. If this interval is nonempty, then extend $f$ to a monotone function $\mathbb{R}\rightarrow \mathbb{R}$ by attaching lines of slope $\pm 1$ outside of the interval, and add the constraint $x\ge f(y)$. 

Similarly, let $g$ be a monotone function $\mathbb{R}\rightarrow \mathbb{R}$ that is equal to $\frac{ay+b}{cy+d}$ when $cy+d> 0$ and $cy+d\le ay+b\le 0$ (assuming that such values of $y$ exist), and add the constraint $x\le g(y)$. The constraints $x\ge f(y)$ and $x\le g(y)$ are equivalent to \eqref{eq:visibilitycondition}. 

We want the constraints to have form $x\le f(y)$ for an increasing function $f$. If $ad-bc$ is positive, then $f$ and $g$ are increasing, and we can write $x\ge f(y)$ and $-x\le \widetilde{f}(-y)$, where $\widetilde{f}(z)=-f(-z)$. If $ad-bc$ is negative, then $f$ and $g$ are decreasing, and we can write $x\ge f(y)$ as $-x\le -f(y)$ and $x\le g(y)$ and $x\le \hat{g}(-y)$ where $\hat{g}(z)=g(-z)$. If $ad-bc=0$, then $\frac{ay+b}{cy+d}$ doesn't depend on $y$, and we adjust the range of $x$ appropriately. 

All of the necessary constraints can be written in a similar way, and they will all be piecewise-fractional-linear. We obtain an instance of $\mathcal{M}-2\text{SAT}$.
\end{proof}

So if we can show that $\mathcal{M}-2\text{SAT}$ is in NP, then we will show that the boundary-boundary art-gallery problem is in NP (the certificate consists of values of the guesses needed in the proof of \Cref{thm:BBreduction} and a certificate for the $\mathcal{M}-2\text{SAT}$ instance generated). Most of the rest of this paper will be devoted to proving that $\mathcal{M}-2\text{SAT}$ is in NP. This is not immediately clear, since irrational values are sometimes required. For example, the instance:

\[x\le \frac{x+p}{x+1}, x\le \frac{x-p}{-x+1}\]

\noindent is satisfiable for $p$ positive, but for $0<p<1$ the only solutions are $x=\pm \sqrt{p}$, which can be irrational. What we will show is that if an instance is satisfiable, then there is a solution where each variable is of form $p+q\sqrt{r}$, where $p$, $q$ and $r$ are rational numbers with at most polynomially many bits.

\section{The continuous inference system}\label{sec:ctsinference}

In this section, we construct an inference system for $\mathcal{F}-2\text{SAT}$ instances and prove that it is refutation complete.

Throughout this section, we work with $n$ real variables $v_1, \dots, v_n$, where $v_i$ is restricted to a compact interval $\text{range}(v_i)\subseteq \mathbb{R}$. In this section, a \emph{formula} is a conjunction \emph{clauses} that are each a conjunction of at most $2$ terms of form $v_i\le c$ or $v_i\ge c$ for some $c\in \mathbb{R}$. The empty clause $\emptyset$ is unsatisfiable. A \emph{continuous formula} is a conjunction of finitely many \emph{constraints}, where a constraint is of form $x\le f(y)$ for some literals $x$ and $y$, where $f$ is a continuous strictly increasing bijection $\mathbb{R}\rightarrow \mathbb{R}$.

\begin{theorem}\label{thm:LatticeFinite}(Beckert, Hahnle and Manya, \cite{Monotone2SAT})
Let $N$ be a lattice (that is, a partially ordered set where every pair of elements has a unique least upper bound and greatest lower bound), let $v_1,\dots, v_n$ be variables taking values in $N$, and let $\Gamma$ be a conjunction of (finitely many) clauses each containing literals of form $x_i\le c$ or $x_i\ge c$ for $c\in N$. If $\Gamma$ is unsatisfiable, then the empty clause can be derived from $\Gamma$ using the binary resolution rule:

\[\frac{(x_1\le a)\vee T_1\quad (x_1\ge b)\vee T_2}{T_1\vee T_2}\text{ where }a<b\]
\end{theorem}

The authors of \cite{Monotone2SAT} restricted to the case where $N$ is finite, but the proof works for $N$ infinite as long as the number of clauses is finite (we can just restrict to the finite sublattice generated by the elements appearing in the $\Gamma$). In the case where $N\subseteq \rr$ is a compact interval, we now show how to use a compactness argument to extend to the case where the number of clauses is infinite.

\begin{lemma}\label{lem:Compactness}
For real variables $v_1, \dots, v_n$, let $\text{range}(v_i)\subseteq \rr$ be a compact interval. Suppose $\Xi$ is an unsatisfiable formula consisting of a (possibly infinite) conjunction of disjunctions of terms $v_i\le c$ or $v_i\ge c$ with $c\in \text{range}(v_i)$. Then $\emptyset$ can be derived from $\Xi$ with (finitely many) applications of binary resolution.
\end{lemma}

\begin{proof}
Let $D\subset \rr^n$ be the set of values $(v_1,\dots, v_n)$ with each $v_i$ in $\text{range}(v_i)$. Since each $\text{range}(v_i)$ is compact, $D$ is a product of compact spaces, so is compact.

For each clause $C$ in $\Xi$, the set of points $S(C)$ in $D$ satisfying $C$ is closed. Since $\Xi$ is unsatisfiable, the intersection:

\[\bigcap\{S(C): C\text{ a clause in }\Xi\}\]

\noindent is empty. So:

\[\bigcup\{D\setminus S(C) : C\text{ a clause in }\Xi\}=D\setminus\bigcap\{S(C): C\text{ a clause in }\Xi\}=D\]

Since each $D\setminus S(C)$ is open and $D$ is compact, there is a finite set of the $D\setminus S(C)$ that cover $D$, and so the intersection of these $S(C)$ is empty. That is, there is a finite subformula $\Xi'$ of $\Xi$ that is not satisfiable.

By \Cref{thm:LatticeFinite}, binary resolution can be used to derive $\emptyset$ from $\Xi'$. Since $\Xi'$ is a subformula of $\Xi$, this gives a (finite) derivation of $\emptyset$ from $\Xi$.
\end{proof}

We now write $x\le f(y)$ in terms of infinitely many clauses involving terms of form $x\le c$ or $y\ge d$.

\begin{lemma}\label{lem:FunctionClauses}
Let $f:\rr\rightarrow\rr$ be a continuous increasing bijection. Values $x$ and $y$ satisfy $x\le f(y)$ if and only if they satisfy $(x\le f(c))\vee (y\ge c)$ for each $c\in \rr$.
\end{lemma}

\begin{proof}
First, suppose $x>f(y)$. By continuity of of $f$, there is some $c>y$ such that $x>f(c)$. So $x$ and $y$ do not satisfy $(x\le f(c))\vee (y\ge c)$.

Now suppose $x\le f(y)$ and let $c\in \rr$. If $y<c$, then $f(y)<f(c)$, so $x\le f(c)$. So at least one of $(x\le f(c))$ or $(y\ge c)$ is satisfied.
\end{proof}

We next want to define an infinite set of clauses $C(x, y, f)$ associated to the constraint $x\le f(y)$. In order to eventually use \Cref{thm:LatticeFinite}, we don't want to have a clause of form $x\le c$ or $x\ge c$ unless $c\in \text{range}(x)$. So we take all the clauses $(x\le f(c))\vee (y\ge c)$ whenever $f(c)$ is in $\text{range}(x)$ and $c$ is in $\text{range}(y)$. So $C(x, y, f)$ contains the clauses:

\[\{(x\le f(c))\vee (y\ge c) : c\in \text{range}(y)\cap f^{-1}(\text{range}(x))\}\]

We also need some additional clauses corresponding to the endpoints of of $\text{range}(x)$ and $\text{range}(y)$. If $f(\text{max}(\text{range}(y)))\in \text{range}(x)$, then we include $x\le f(\text{max}(\text{range}(y)))$. If $ f^{-1}(\text{min}(\text{range}(x)))\in \text{range}(y)$, then we include $y\ge f^{-1}(\text{min}(\text{range}(x)))$. 

If $f(\text{max}(\text{range}(y)))<\text{min}(\text{range}(x))$ or $f^{-1}(\text{min}(\text{range}(x)))>\text{max}(\text{range}(y))$, then we add the empty clause $\emptyset$ to $C(x, y, f)$, since the constraint $x\le f(y)$ is unsatisfiable for $x\in \text{range}(x)$ and $y\in \text{range}(y)$. 

\begin{lemma}\label{lem:ConversionCorrectness}
Values $x$ and $y$ with $x\in \text{range}(x)$ and $y\in \text{range}(y)$ satisfy $x\le f(y)$ if and only if they satisfy all the constraints in $C(x, y, f)$.
\end{lemma}

\begin{proof}
First suppose $x$ and $y$ satisfy $x\le f(y)$. By \Cref{lem:FunctionClauses}, all the $(x\le f(c))\vee (y\ge c)$ clauses are satisfied. Since $y\in \text{range}(y)$, we have $x\le f(\text{max}(\text{range}(y)))$. Since $x\in \text{range}(x)$, we have $f(y)\ge \text{min}(\text{range}(x))$, so $y\ge f^{-1}(\text{min}(\text{range}(x)))$.

Now suppose that $x$ and $y$ are such that $x>f(y)$. By \Cref{lem:FunctionClauses}, $(x\le f(c))\vee (y\ge c)$ fails for some $c\in \rr$. If this value of $c$ is in $\text{range}(y)\cap f^{-1}(\text{range}(x))$, then we are done. 

Suppose $(x\le f(c))\vee (y\ge c)$ fails for some $c\notin \text{range}(y)$. Since $y<c$, we have $c>\text{max}(\text{range}(y))$. Since $x>f(c)$, $x>f(\text{max}(\text{range}(y)))$, which fails to satisfy $x\le f(\text{max}(\text{range}(y)))$.

Similarly, if $(x\le f(c))\vee (y\ge c)$ fails for some $f(c)\notin \text{range}(x)$, then $f(c)<\text{min}(\text{range}(x))$ so $y<f^{-1}(\text{min}(\text{range}(x)))$.
\end{proof}

If $\Gamma$ is a continuous formula, then we define the associated formula $\Xi$ to be the formula consisting of the conjunction of all the clauses in $C(x, y, f)$ for each $x, y, f$ such that $x\le f(y)$ is a constraint in $\Gamma$.

\subsection{Inference with function composition}

We define three \emph{continuous inference rules}. These are:

\[\frac{x\le f(y) \quad y\le g(z)}{x\le f(g(z))}\]
\[\frac{x\le f(y)}{\emptyset}\text{ where }f(\text{max}(\text{range}(y)))<\text{min}(\text{range}(x))\]
\[\frac{x\le f(-x)\quad -x\le g(x)}{\emptyset}\text{ where }\exists c: f(-c)<c\text{ and }g(c)<-c\]

\begin{lemma}\label{lem:InferenceValidity}
Let $\Gamma$ be a continuous formula and let $\Gamma'$ be a continuous formula that can be obtained from $\Gamma$ by one of the continuous inference rules. Then $\Gamma$ is satisfiable if and only if $\Gamma'$ is.
\end{lemma}

\begin{proof}
$\Gamma'$ contains all the clauses in $\Gamma$, so $\Gamma'$ is unsatisfiable if $\Gamma$ is.

If $\Gamma$ is satisfiable, then there is some assignment of the variables satisfying all the clauses. We want to show that the same variable assignment satisfies $\Gamma'$. If $x, y$ and $z$ satisfy $x\le f(y)$ and $y\le g(z)$, then by monotonicity of $f$, $f(y)\le f(g(z))$, so $x\le f(g(z))$. So the first inference rule is valid.

If $x\le f(y)$ and $y\in \text{range}(y)$, then $y\le \text{max}(\text{range}(y))$ so $x\le f(\text{max}(\text{range}(y)))$. So if $f(\text{max}(\text{range}(y)))<\text{min}(\text{range}(x))$ then $x$ is not in $\text{range}(x)$. So if $\Gamma$ is satisfiable then the second inference rule can never be applied.

If $x\le f(-x)$ and $-x\le g(x)$ and there is some $c$ such that $f(-c)<c$ and $g(c)<-c$, then $0\le f(-x)-x$ and $0\le g(x)+x$. The function $f(-x)-x$ is increasing in $-x$ and $g(x)+x$ is increasing in $x$, so $-x>-c$ and $x>c$, which is a contradiction. This case is shown in \Cref{fig:ThirdRule}. So if $\Gamma$ is satisfiable then the third inference rule can never be applied.

\begin{figure}
\centering
\includegraphics[page=2]{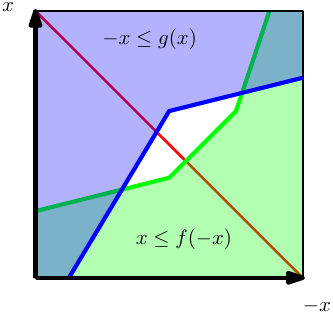}
\caption{The case where the third inference rule is needed. There appear to be points that satisfy both constraints, but none where $x=-(-x)$. The first two inference rules don't detect the relationship between $x$ and $-x$ as literals, so aren't powerful enough to refute this case. }
\label{fig:ThirdRule}
\end{figure}
\end{proof}

If $f:\rr\rightarrow \rr$ is a strictly increasing continuous bijection, then $f$ has a well-defined continuous inverse. The constraint $x\le f(y)$ could equivalently be written $f^{-1}(x)\le y$. Define $\widetilde{f}:\rr \rightarrow \rr$ by $\widetilde{f}(x)=-f^{-1}(-x)$. So $\widetilde{f}$ is a strictly increasing continuous bijection and $x\le f(y)$ is equivalent to $-y\le \widetilde{f}(-x)$. Note that $C(x, y, f)$ contains exactly the same clauses as $C(-y, -x, \widetilde{f})$. We say that a continuous formula $\Gamma$ is \emph{symmetric} if, for every constraint of the form $x\le f(y)$, the constraint $-y\le \widetilde{f}(-x)$ also appears.

\begin{lemma}
If $\Gamma$ is symmetric and $\Gamma'$ can be derived from $\Gamma$, then a symmetric formula containing $\Gamma'$ can be obtained from $\Gamma'$.
\end{lemma}

\begin{proof}
\[\widetilde{f\circ g}=\widetilde{g}\circ\widetilde{f}\]
\end{proof}

We can now state the main result of this section:

\resolutioncompleteness*


To prove \Cref{lem:ResolutionCompleteness}, we will use \Cref{lem:Compactness} to show that if $\Gamma$ is unsatisfiable, then the empty clause can be derived from the constraints in the $C(x, y, f)$ by binary resolution. If a clause $C$ can be derived from clauses in $C(x, y, f)$ and $C(y, z, g)$, then we would like to argue that $C$ is contained in $C(x, z, f\circ g)$. However, this is not quite true.

For a clause $C$, we say that another clause $C'$ is \emph{stronger than} $C$ if the set of values satisfying $C'$ is contained in the set of values satisfying $C$. Roughly speaking, the strategy to prove \Cref{lem:ResolutionCompleteness} is to show that, if a clause $C$ can be derived from $C(x, y, f)$ and $C(y, z, g)$, then a stronger clause $C'$ is contained in $C(x, z, f\circ g)$.

However, this creates a problem. The inference system in \cite{Monotone2SAT} implicitly allows a clause of form $(x\le c)\vee (x\le c)$ to be contracted to $x\le c$. However, the stronger clause $(x\le c)\vee (x\le d)$ for $d<c$ can't be contracted in this way. This is a problem for us because it means that derivations behave badly with respect to strengthening clauses. A derivation starting with a formula including $(x\le c)\vee (x\le c)$ might need fewer steps than one where $(x\le c)\vee (x\le c)$ is replaced by $(x\le c)\vee (x\le d)$ for $d<c$. We instead use an inference system with \emph{strong contractions}, where $(x\le c)\vee (x\le d)$ can be contracted to $x\le c$ whenever $c\ge d$.

\begin{lemma}\label{lem:ProofMonotonicity}
Let $\Gamma$ be a formula that can be refuted in $k$ steps of binary resolution, allowing strong contractions. If $\Gamma'$ is obtained from $\Gamma$ by replacing a clause $C$ with a stronger clause $C'$, then there is a derivation of the empty clause from $\Gamma'$ requiring the same number of steps, again allowing strong contractions.
\end{lemma}

\begin{proof}
Straightforward.
\end{proof}

We are now ready to prove \Cref{lem:ResolutionCompleteness}. 

\begin{proof}[Proof of \Cref{lem:ResolutionCompleteness}]
Let $\Xi$ be formula consisting of the disjunction of all the $C(x, y, f)$ for each $x\le f(y)$ that occurs in $\Gamma$. By \Cref{lem:ConversionCorrectness}, $\Xi$ is unsatisfiable, so by \Cref{lem:Compactness}, there is a derivation of the empty clause from $\Xi$. We will show that if a clause $C$ that can be obtained from $\Xi$ by binary resolution, then we can derive some $x\le f(y)$ from $\Gamma$ such that $C(x, y, f)$ contains a clause stronger than $C$. By induction and \Cref{lem:ProofMonotonicity}, this will show that the empty clause can be derived from $\Gamma$. 

Each clause in $\Xi$ has form either $(x\le f(c))\vee (y\ge c)$, $x\le f(\text{max}(\text{range}(y)))$, or $x\ge f^{-1}(\text{min}(\text{range}(y)))$ for some constraint $x\le f(y)$ in $\Gamma$. The clause $x\ge f^{-1}(\text{min}(\text{range}(y)))$ can alternately be written as $-x\le \widetilde{f}(\text{max}(\text{range}(-y)))$. Since $\Gamma$ is symmetric, we can always use the later form of this clause.

Throughout the proof, we will view terms of form $x\le c$ and $-x\ge -c$ as being interchangeable. We will choose to write terms in either $\le$ or $\ge$ form where convenient. 

Let $C$ be a clause that can be derived directly from $\Xi$ by binary resolution. There are three cases to consider, depending on whether $C$ has $0$, $1$, or $2$ terms.

Suppose $C$ has two terms, write $C=(x\ge c)\vee (z\le d)$ where $x$ and $z$ are literals and $c$ and $d$ are constants. In order to produce $C$ by binary resolution, we need clauses $(y\le a)\vee (x\ge c)$ and $(y\ge b)\vee (z\le d)$ for some literal $y$ and constants $b>a$. 

A clause of form $(y\le a)\vee (x\ge c)$ can occur in some $C(y, x, f)$ as $(y\le f(c_1))\vee (x\ge c)$ (so $a=f(c)$), or in $C(-y, -x, f)$ as $(-x\le f(-a))\vee (-y\ge -a)$ (so $c=f(-a)$). Since $\Gamma$ is symmetric, we can assume that it is the first case; if the clause comes from $C(-y, -x, f)$, then there is an equivalent clause in $C(x, y, \widetilde{f})$. 

Similarly, we can suppose that $(y\ge b)\vee (z\le d)$ comes from some $C(z, y, g)$, where $d=g(b)$. 

Now the constraints $y\le f(x)$ and $z\le g(y)$ yield $v\le g(f(x))$. We have $d=g(b)\ge g(a)=g(f(c))$ so $C(z, x, g\circ f)$ contains $(z\le g(f(c)))\vee (x\ge c)$, which is as least as strong as $(x\ge c)\vee (z\le d)$.

Now suppose that $C$ has $1$ term, write $C=(x\le a)$. The only way to way to obtain this clause by binary resolution is to start with one clause with two terms and one clause with one term. The clause with two terms can be written $(x\le f(c))\vee (y\ge c)$, coming from a constraint $x\le f(y)$, where $f(c)=a$.

Clauses of form $x\le f(\text{max}(\text{range}(y)))$ appear directly in $\Xi$, but there are also contractions of a clauses $(y\le f(c))\vee (-y\ge c)$ contained in a constraint like $C(y, -y, f)$.

Suppose the single-term clause is of form $y\le g(\text{max}(\text{range}(z)))$, so $g(\text{max}(\text{range}(z)))<c$ in order to apply resolution. Now $\Gamma$ has constraints $x\le f(y)$ and $y\le g(z)$, so we can derive $x\le f(g(z))$, and $C(x, z, f\circ g)$ contains $x\le y\le f(g(\text{max}(\text{range}(z))))$. Since $f(g(\text{max}(\text{range}(z))))\le f(c)=a$, this is at least as strong as $C$. 

Now suppose that the single-term clause $y\le b$ appears as a contraction of a clause like $(y\le g(d))\vee (-y\ge d)$ from $C(y, -y, g)$. By symmetry, we can suppose that $g(d)\ge -d$ (since otherwise $-\widetilde{g}(-g(d))\ge -g(d)$ and we could use $C(y, -y, \widetilde{g})$ instead), so the clause contracts to $(y\le g(d))$. Since we can apply resolution, $g(d)<c$. 

From $x\le f(y)$ and $y\le g(-y)$ we derive $x\le f(g(-y))$. From$x\le f(g(-y))$ and $-y\le \widetilde{f}(-x)$ we derive $x\le f(g(\widetilde{f}(-x)))$. Now $\widetilde{f}(-f(c))=-c<-g(d)$, so $g(\widetilde{f}(-f(c)))\le g(-g(d))\le g(d)<c$. So $f(g(\widetilde{f}(-f(c))))\le f(c)$, and so $(x\le f(g(\widetilde{f}(-f(c)))))\vee (-x\ge -f(c))$ contracts to $-x\ge -f(c)$, which is equivalent to $x\le f(c)$. 

Finally, we consider the case where $C=\emptyset$. This must have come from the resolution of two clauses each with $1$ term. There are three cases to consider, depending on whether $0$, $1$, or $2$ of the inputs came from a contraction.

If neither input came from a contraction, then we started with clauses of form $x\le f(\text{max}(\text{range}(y)))$ and $-x\le g(\text{max}(\text{range}(z)))$ where $f(\text{max}(\text{range}(y)))<-g(\text{max}(\text{range}(z)))$. From $x\le f(y)$ and $-z\le \widetilde{g}(x)$, we can derive $-z\le \widetilde{g}(f(y))$. Since $f(\text{max}(\text{range}(y)))<-g(\text{max}(\text{range}(z)))$, we have $\widetilde{g}(f(\text{max}(\text{range}(y))))<-\text{max}(\text{range}(z))=\text{min}(\text{range}(-z))$. So $\emptyset$ can be derived from $-z\le \widetilde{g}(f(y))$.

If one of the inputs comes from a contraction, then we started with clauses $x\le f(\text{max}(\text{range}(y)))$ and $(-x \le g(-c))\vee (x\ge -c)$. By symmetry we can suppose that $g(-c)\le c$. Now $-x\le g(f(y))$ can be derived from $-x\le g(x)$ and $x\le f(y)$. Since $f(\text{max}(\text{range}(y)))<-c$, $g(f(\text{max}(\text{range}(y))))\le c$. So we now have $x\le f(y)$ and $-x\le g(f(y))$ where $f(\text{max}(\text{range}(y)))<-g(f(\text{max}(\text{range}(y))))$, which reduces to the previous case.

Finally, we have the case where both clauses come from contractions. We can write these clauses $(x\le f(-a))\vee (-x\ge -a)$ and $(-x\le g(b))\vee (x\ge b)$, where $f(-a)\le a$ and $g(b)\le -b$ (by symmetry). So the clauses contract to $x\le a$ and $x\ge b$, with $b>a$. Choose some $c$ in $(a, b)$. The function $f(-x)-x$ is strictly increasing in $-x$ and $g(x)+x$ is strictly increasing in $x$. Since $-c<-a$, $c<b$ and $f(-a)-a\le 0$ and $g(b)+b\le 0$, we have $f(-c)-c<0$ and $g(c)+c<0$, so $f(-c)<c$ and $g(c)<-c$. So $\emptyset$ can be derived from $x\le f(-x)$ and $-x\le g(x)$.
\end{proof}

\section{Bounded-depth refutations}\label{sec:logdepth}

In general, there is no way to bound the length of the refutation produced by \Cref{lem:ResolutionCompleteness}. However, if a refutation is too large, then we might be able to simplify it. For example, a refutation might generate many different constraints of form $x\le f(y)$, in which case we can replace constraints $x\le f_1(y),\dots, x\le f_k(y)$ with a single constraint $x\le \text{min}(f_1, \dots, f_k)(y)$. 

If a constraint of form $x\le f(x)$ appears, then we can obtain $x\le (f\circ\dots\circ f)(x)$. We may need many compositions of $f$ in order to obtain a strong enough constraint, so this is another way that the length of a refutation can become very large. In these cases, we can compute (roughly) the limit of $f^k$ as $k$ goes to $\infty$. As we show in \Cref{sec:infinitecomposition}, the constraint $x\le f^{\infty}(x)$ and the tautology $x\le x$ are together stronger than any constraint of form $x\le f^k(x)$, and so we are saved from having to compute an arbitrarily large
number of compositions of $f$.

In this section, we show that the addition of $k$-way minimum operations and infinite compositions allow us to produce a refutations where the depth of arithmetic operations required is at most logarithmic. For every pair of literals $x$ and $y$, this refutation tries to calculate (roughly) the function $\text{min}_{x\rightarrow y}$ representing the infimum over all $f$ such that the constraint $x\le f(y)$ can be derived. 

Given an $\mathcal{F}-2\text{SAT}$ instance $\Gamma$, let $G$ be a graph with vertices representing literals $\pm v_i$ and directed edges representing constraints, so that an edge from $x$ to $y$ represents a constraint of form $x\le f(y)$ in $\Gamma$. The constraints $x\le f(y)$ that can be derived from $\Gamma$ correspond to paths from $x$ to $y$ in $G$. 

If $x\le f(y)$ is a constraint that can be derived from $\Gamma$ using the composition, minimum, and infinite composition operations, then the sequence of operations used to define $f$ represents (in some sense) a regular expression matching paths in $G$, and $f$ is the infimum of the constraints represented by the paths matched by the regular expression. Given vertices $x$ and $y$ in $G$, it is possible (see Brzozowski and McCluskey \cite{GraphRegex}) to construct a regular expression matching all directed paths from $x$ to $y$. However, the regular expression produced can be as long as $n^n$, which is not useful for us.

So instead, we construct a regular expression that matches the paths from $x$ to $y$ of length at most $2n$, so in particular it matches any non-repeating paths. We let $f_{x\rightarrow y}$ be the function so that this regular expression constructs the constraint $x\le f_{x\rightarrow y}(y)$. In \Cref{sec:paths}, we show that, for a real number $c$, $\text{min}_{x\rightarrow y}(c)$ is either $f_{x\rightarrow y}(c)$ or $f_{x\rightarrow z}(d)$ for some literal $z$, where $d$ is an attracting fixed point of $f_{z\rightarrow z}$. In the case of $\mathcal{M}-2\text{SAT}$, this will give a bound on the complexity of coordinates required for a satisfying assignment, proving NP-membership for $\mathcal{M}-2\text{SAT}$.

For $\mathcal{M}-2\text{SAT}$, the refutations considered in \Cref{sec:paths} can be explicitly constructed in time $n^{\mathcal{O}(\log(n))}$, allowing us to solve $\mathcal{M}-2\text{SAT}$ in quasi-polynomial time. In particular, $\mathcal{M}-2\text{SAT}$ isn't NP-hard unless the exponential-time hypothesis is false. It seems to be difficult to perform these calculations in polynomial time, since the regular expression represented by $f_{x\rightarrow y}$ is a monotone formula for solving directed st-connectivity. Such a formula must have size at least $n^{\Omega(\log(n))}$ (see Karchmer and Wigderson \cite{STLowerBound}).

\subsection{Infinite compositions}\label{sec:infinitecomposition}

If a constraint of form $x\le f(x)$ appears, then we can obtain $x\le (f\circ\dots\circ f)(x)$. We may need many compositions of $f$ in order to obtain a strong enough constraint. Our strategy to handle this case is to compute (roughly) the limit of $f^k(x)$ as $k$ goes to $\infty$. As we show below, the constraint $x\le f^{\infty}(x)$ and the tautology $x\le x$ are together stronger than any constraint of form $x\le f^k(x)$, so we are saved from having to compute an arbitrarily large number of compositions of $f$.

Given a function $f$, we now define precisely what we mean by the function $f^{\infty}$. Our precise definition is somewhat subtle, because we want $f^{\infty}$ to be piecewise-constant. This might fail if $f$ has infinitely many fixed points, for example if $f$ is a piecewise linear function where one piece is $f(x)=x$.

If $f:\rr \rightarrow \rr$ is a strictly increasing continuous bijection, then we define $f^{\infty}:\rr \rightarrow \rr \cup \{\pm \infty\}$ by:

\[f^{\infty}(x)=\inf\{y\ge \sup\{z\le x : z\le f(z)\} : y>f(y)\}\]

\noindent here taking $\inf\{\emptyset\}=+\infty$ and $\sup\{\emptyset\}=-\infty$.

Equivalently, we could define:

\[f^{\infty}(x)=\lim_{\epsilon\rightarrow 0^+}\lim_{k\rightarrow\infty}(f+\epsilon)^{k}(x)\] 

\noindent but the first definition will be easier to work with. \Cref{fig:InfiniteComposition} illustrates this construction. 

\begin{figure}
\centering
\includegraphics[page=3]{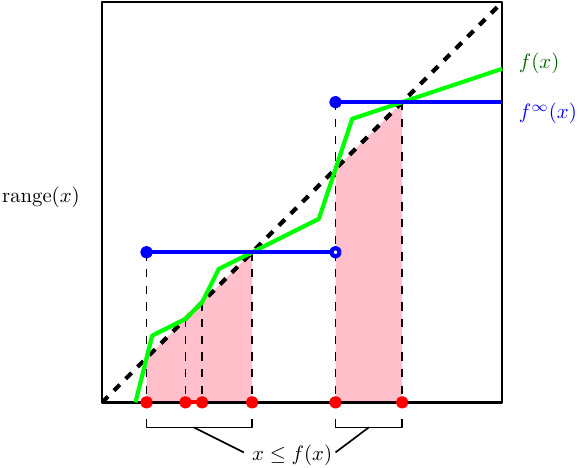}
\caption{Illustration of $f^{\infty}$ }
\label{fig:InfiniteComposition}
\end{figure}

We now prove some useful properties of this construction. 

\begin{lemma}\label{lem:InfiniteBound}
If $f(x)<x$, then $f^{\infty}(x)\le f(x)$. 
\end{lemma}

\begin{proof}
If $z$ is such that $f(x)< z\le x$ then $f(z)\le f(x)<z$. So $\sup\{z\le x : z\le f(z)\}\le f(x)$. Since $f(x)<x$, $f(f(x))<f(x)$, so $f^{\infty}(x)\le f(x)$.
\end{proof}

\begin{lemma}\label{lem:LimitValidity}
$x\le f(x)$ if and only if $x\le f^{\infty}(x)$.
\end{lemma}

\begin{proof}
If $x\le f(x)$, then $\sup\{z\le x : z\le f(z)\}=x$, so $f^{\infty}(x)\ge x$.

If $x>f(x)$, then by \Cref{lem:InfiniteBound}, $f^{\infty}(x)\le f(x)$, so $x>f^{\infty}(x)$.
\end{proof}

\begin{lemma}\label{lem:LimitFiniteness}
For all $x\in \rr$, $f^{\infty}(x)=\pm\infty$ or $f^{\infty}(x)$ is a fixed point of $f$ and $f^{\infty}(x)+\epsilon$ is not a fixed point of $f$ for a sequence of $\epsilon>0$ limiting to zero. 

Also $f^{\infty}\circ f^{\infty}=f^{\infty}$.
\end{lemma}

\begin{proof}
When $f^{\infty}(x)$ is finite, the definition tells us that $f^{\infty}(x)+\epsilon>f(f^{\infty}(x)+\epsilon)$ for arbitrarily small $\epsilon\ge 0$ and that $f^{\infty}(x)-\delta\le f(f^{\infty}(x)-\delta)$ for arbitrarily small $\delta\ge 0$. By continuity of $f$, $f^{\infty}(x)=f(f^{\infty}(x))$. 

So $f^{\infty}(x)+\epsilon>f(f^{\infty}(x)+\epsilon)$ for arbitrarily small $\epsilon>0$ (since it does not hold when $\epsilon=0$). 

If $x$ is a fixed point of $f$ and $x+\epsilon>f(x+\epsilon)$ for arbitrarily small $\epsilon>0$, then $\sup\{z\le x:z\le f(z)\}=x$ and $\inf\{y\ge x:y>f(y)\}=x$, so $f^{\infty}\circ f^{\infty}=f^{\infty}$.
\end{proof}

Those fixed points of $f$ that occur as values of $f^{\infty}$ are called \emph{attracting points} of $f$.

\subsection{Compressing paths}\label{sec:paths}

Given two literals $x$ and $y$, any function $f$ such that $x\le f(y)$ can be derived from a continuous formula $\Gamma$ can be written $f=f_1\circ\dots\circ f_k$ for functions $f_i$ from the original problem instance. We would like to compute the minimum (that is, infimum) of the possible functions $f$ such that $x\le f(y)$ can be derived.

Given a continuous formula $\Gamma$, we say a \emph{path} $p$ from a literal $x$ to a literal $y$ is a list of constraints $x=x_1\le f_1(x_2), \dots, x_k\le f_k(x_{k+1}=y)$. A path is non-repeating if $x_i\ne x_j$ for $i\ne j$ unless $\{i, j\}=\{1, k+1\}$ (so a path from $x$ to itself can be non-repeating by this definition). The resolution of a path is the constraint $x\le (f_1\circ \dots \circ f_k)(y)$. We will sometimes write $p(y)$ for $(f_1\circ \dots \circ f_k)(y)$. The interior vertices of a path are the vertices $x_2, \dots, x_{k-1}$. 

The number of paths from $x$ to $y$ is infinite in general, but we hope to use infinite compositions in order to reduce this to a finite size. But there can still be roughly $n!$ non-repeating paths between a given pair of literals. 

We now introduce the minimum operation. If we have constraints $x\le f_1(y), \dots, x\le f_k(y)$, then we can replace them by $x\le \text{min}(f_1, \dots, f_k)(y)$. Note that an expression like $\text{min}(f_1, g_1)\circ \dots\circ \text{min}(f_n, g_n)$ can represent the minimum of an exponential number of paths in only polynomial size. 

Given a continuous formula $\Gamma$, we produce a new continuous formula $H(\Gamma)$ by the following process:

\begin{itemize}
    \item For each pair of constraints of form $x\le f(y)$ and $y\le g(z)$ with $y$ not being the same literal as $x$ or $z$, add the clause $x\le f(g(z))$.
    \item For every pair of (not necessarily distinct) literals $x$ and $y$, write $\{x\le f_i(y):i=1,\dots, k\}$ for the set of constraints from $x$ to $y$ in $\Gamma$, and replace them with the single constraint $x\le \min_i(f_i)(y)$.
\end{itemize}

\begin{lemma}\label{lem:HValidity}
$H(\Gamma)$ is satisfiable if an only if $\Gamma$ is satisfiable. 
\end{lemma}

\begin{proof}
Whenever a constraint of form $x\le f(y)$ appears in $\Gamma$, a constraint $x\le g(y)$ appears in $H(\Gamma)$ where $g(c)\le f(c)$ for all $c\in \rr$. So if $H(\Gamma)$ is satisfiable then $\Gamma$ is.

If $x\le f_i(y)$ for $i=1, \dots, k$, then $x\le \min_i(f_i)(y)$. If $x\le f(y)$ and $y\le g(z)$, then $x\le f(g(z))$. So if $\Gamma$ is satisfiable, then the same assignment satisfies the constraints in $H(\Gamma)$. 
\end{proof}

Write $H^2(\Gamma)=H(H(\Gamma))$ and write $H^k(\Gamma)$ for the result after applying this process $k$ times.

\begin{lemma}\label{lem:HCompleteness}
Let $\Gamma$ be a continuous formula with $n$ variables. Suppose there is a non-repeating path $p$ from a literal $x$ to a literal $y$ in $\Gamma$. If $m\ge \log_2(n)+1$ then $H^m(\Gamma)$ contains a constraint $x\le f(y)$ with $f(c)\le p(c)$ for all $c\in \rr$.
\end{lemma}

\begin{proof}
By induction, we show that for non-repeating path $p$ from $x$ to $y$ in $\Gamma$ of length at most $2^k$, there is a constraint $C(x, y, f)$ in $H^k(\Gamma)$ that is stronger than the resolution of $p$. This is clear when $k=0$.

Suppose this is true for some $k$ and $p$ is a path from $x$ to $y$ of length at most $2^{k+1}$. Then $p$ can be split into non-repeating paths $p_1$ from $x$ to $z$ and $p_2$ from $z$ to $y$, each of length at most $2^k$. By induction, $H^k(\Gamma)$ contains constraints $x\le f(z)$ and $z\le g(y)$ with $f\le p_1$ and $g\le p_2$. So $f(g(c))\le p(c)$ for $c\in \rr$. Since $z$ is not the same literal as $x$ or $y$, $H^{k+1}(\Gamma)$ contains a constraint $x\le h(y)$ where $h\le f\circ g$.

A non-repeating path has length at most $2n$, so $H^{m}(\Gamma)$ contains constraints for all such paths if $m\ge \log_2(n)+1$.
\end{proof}

Now let $V(\Gamma)$ be the continuous formula obtained from $H^{\lceil\log_2(n)+1\rceil}(\Gamma)$ by replacing every constraint of form $x\le f(x)$ with $x\le f^{\infty}(x)$. We say that a path in $V(\Gamma)$ is \emph{tight} if it alternates between loops and edges that aren't loops. 

\begin{lemma}\label{lem:VCompleteness}
Suppose there is a path $p$ from $x$ to $y$ in $\Gamma$. Then for each $c\in \rr$, there is a tight path $q$ from $x$ to $y$ in $V(\Gamma)$ such that $q(c)\le p(c)$.
\end{lemma}

\begin{proof}
If $p$ is non-repeating then we are done by \Cref{lem:HCompleteness}.

Now induct on the length of $p$. A path of length $1$ can never repeat. Suppose that $p$ repeats a literal $z$. Split $p$ at the first and second points where it hits $z$, so $p$ splits into sub-paths $\{p_1, p_2, p_3\}$ where $p_2$ is a non-repeating cycle. 

Let $c\in \rr$. First suppose that $(p_2\circ p_3)(c)\ge p_3(c)$. In this case, $(p_1\circ p_2 \circ p_3)(c)\le (p_1\circ p_3)(c)$. The path $\{p_1, p_3\}$ is a path shorter than $p$, so $p_1(p_3(c))$ is at most $q(c)$ where $q$ is a tight path.

Now suppose that $p_2(p_3(c))<p_3(c)$. Since $p_3$ is shorter than $p$, $q_3(c)\le p_3(c)$ for some tight path $q_3$. By \Cref{lem:HCompleteness}, there is a loop $\ell$ in $H^{\lceil\log_2(m)+1\rceil}(\Gamma)$ where $\ell(p_3(c))\le p_2(p_3(c))$. Since $p_2(p_3(c))<p_3(c)$, $\ell(p_3(c))<p_3(c)$, so $\ell^{\infty}(p_3(c))\le p_2(p_3(c))$ by \Cref{lem:InfiniteBound}. So $\ell^{\infty}(q_3(c))\le p_2(p_3(c))$. There is a tight path $q_1$ such that $q_1(\ell^{\infty}(q_3(c)))\le p_1(\ell^{\infty}(q_3(c)))\le p_1(p_2(p_3(c)))$. 

The path $(q_1, \ell^\infty, q_3)$ may not be tight since the paths $q_1$ and $q_3$ can have loops at the ends, so $(q_1, \ell^\infty, q_3)$ might have multiple loops in a row. But $V(\Gamma)$ has only one loop per vertex, and $\ell^\infty\circ \ell^\infty=\ell^\infty$, so if there are multiple loops in a row then we can remove all but one of them to obtain a tight path. 

So for each $c\in \rr$, there is a tight path $q$ in $V(\Gamma)$ with $q(c)\le p(c)$. 
\end{proof}

Write $x\le f_{x\rightarrow y}(y)$ for the constraint from $x$ to $y$ in $H^{\lceil\log_2(n)+1\rceil}(\Gamma)$ (if one exists) and define:

\begin{equation}\label{eq:minxy}
\text{min}_{x\rightarrow y}(y)=\text{inf}\{p(y) : p\text{ is a tight path in }V(\Gamma)\}
\end{equation}

In general, $\text{min}_{x\rightarrow y}$ is not really a minimum but an infimum. However, for suitably ``nice'' functions, it really is a minimum. Recall that a function is analytic if it is $C^{\infty}$ and the Taylor series around any point converges in a positive radius of that point. Polynomials and rational functions are analytic. 

\begin{lemma}\label{lem:minvalues}
Suppose that the functions in $\Gamma$ are piecewise analytic. Then the infimum in \eqref{eq:minxy} is attained. So for $c\in \mathbb{R}$, the value of $\text{min}_{x\rightarrow y}(c)$ is one of:

\begin{itemize}
    \item $f_{x\rightarrow y}(c)$
    \item An attracting point of $f_{x\rightarrow x}$
    \item $f_{x\rightarrow z}(d)$ for some literal $z\ne x$, where $d$ is an attracting point of $f_{z\rightarrow z}$
    \item $\pm \infty$
\end{itemize}
\end{lemma}

\begin{proof}
The values taken by a loop from $z$ to itself in $V(\Gamma)$ are attracting points of $f_{z\rightarrow z}$, so it is clear that the values of a tight path are of the form claimed. So by \Cref{lem:VCompleteness}, the value of $\text{min}_{x\rightarrow y}(c)$ is an infimum of values of the types claimed.

Since the functions in $\Gamma$ are piecewise analytic, the functions in $H^{\lceil\log_2(n)+1\rceil}(\Gamma)$ are piecewise analytic, since they are obtained from functions in $\Gamma$ by compositions and min operations.

To complete the proof, we notice the attracting points of a piecewise analytic function function form a discrete set (that is, each attracting point is separated from all other attracting points by a non-zero distance). Indeed, if $f$ is piecewise analytic and has infinitely many fixed points in arbitrarily small neighborhoods of a point $c$, then $f(x)=x$ in a neighborhood of $c$, and there are no attracting points in the interior of this neighborhood. The infimum of a discrete set is always attained, so this proves the claim. 
\end{proof}

By \Cref{lem:VCompleteness}, we can express the satisfiability of $\Gamma$ in terms of properties of the $\text{min}_{x\rightarrow y}$.

\begin{lemma}\label{lem:KeyLemma}
$\Gamma$ is satisfiable if and only if neither of the following happen:

\begin{itemize}
    \item There are literals $x$ and $y$ such that $\text{min}_{x\rightarrow y}(\text{max}(\text{range}(y)))<\text{min}(\text{range}(x))$
    \item There is a literal $x$ and a $c\in \text{range}(x)$ such that $\text{min}_{x\rightarrow -x}(-c)<c$ and $\text{min}_{-x\rightarrow x}(c)<-c$
\end{itemize}
\end{lemma}

\begin{proof}
If $\Gamma$ is satisfiable, then $V(\Gamma)$ is satisfiable by \Cref{lem:HValidity,lem:LimitValidity}. Each of the cases implies that there is a derivation of $\emptyset$ from $V(\Gamma)$, which can't happen by \Cref{lem:InferenceValidity}.

Suppose $\Gamma$ is unsatisfiable. By \Cref{lem:ResolutionCompleteness}, there is a derivation of $\emptyset$ from $\Gamma$. For every constraint of form $x\le f(y)$ that can be derived from $\Gamma$, the function $f$ is the resolution of a path in $\Gamma$. By \Cref{lem:VCompleteness}, these paths correspond to tight paths $p$ or $p$ and $q$ in $V(\Gamma)$ satisfying one of the conditions.
\end{proof}

\section{Algorithms for $\mathcal{M}-2\text{SAT}$}\label{sec:algorithms}

\begin{theorem}\label{thm:npmembership}
$\mathcal{M}-2\text{SAT}$ is in NP
\end{theorem}

\begin{proof}
Let $\Gamma$ be an instance of $\mathcal{M}-2\text{SAT}$ with $n$ variables where the each constraint in $\Gamma$ is represented with at most $\beta$ bits. Suppose that $\Gamma$ is satisfiable. First, we compute the set of values of $v_1$ that can be extended to satisfying assignments. 

The function $\text{min}_{x\rightarrow y}$ only depend on the functions appearing in $\Gamma$ and not on the ranges $\text{range}(v_i)$. Since $\Gamma$ is satisfiable, the second case in \Cref{lem:KeyLemma} never happens. So by \Cref{lem:KeyLemma} the problem has a satisfying assignment with $v_1=c$ if and only if:

\begin{itemize}
    \item $\text{min}_{-v_1\rightarrow x}(\text{max}(\text{range}(x)))\ge -c$ and $\text{min}_{v_1\rightarrow x}(\text{max}(\text{range}(x)))\ge c$ for each $x$ a literal other than $v_1$ or $-v_1$
    \item $\text{min}_{v_1\rightarrow -v_1}(-c)\ge c$ and $\text{min}_{-v_1\rightarrow v_1}(c)\ge -c$
    \item $\text{min}_{v_1\rightarrow v_1}(c)\ge c$ (equivalently $\text{min}_{-v_1\rightarrow -v_1}(-c)\ge -c$)
\end{itemize}

Note that we don't need to check e.g. $\text{min}_{x\rightarrow v_1}(c)\ge \text{min}(\text{range}(x))$ because this is equivalent to $\text{min}_{-v_1\rightarrow -x}(\text{max}(\text{range}(-x)))\ge -c$.

The set of values of $c$ satisfying all these conditions is a union of closed intervals, where each endpoint of one of the intervals is one of:

\begin{itemize}
    \item $-\text{min}_{-v_1\rightarrow x}(\text{max}(\text{range}(x)))$ for $x$ a literal other than $x$
    \item $\text{min}_{v_1\rightarrow x}(\text{max}(\text{range}(x)))$ for $x$ a literal other than $x$
    \item A point where $\text{min}_{v_1\rightarrow -v_1}(-c)=c$
    \item A point where $\text{min}_{-v_1\rightarrow v_1}(c)=-c$
    \item A locally minimal or locally maximal fixed point of $\text{min}_{v_1\rightarrow v_1}$
\end{itemize}

The constraints in $\Gamma$ are piecewise fractional linear, that is they are piecewise of form:

\[\frac{ax+b}{cx+d}\]

\noindent where $a, b, c$ and $d$ are integers with at most $\beta$ bits each. 

Taking the minimum of some of these constraints increases only the number of pieces and not the complexity of the pieces themselves. Composing two functions of this form causes the bit-complexity of the coefficients to grow by at most a factor of $3$, since:

\[\frac{a_1\left(\frac{a_2x+b_2}{c_2x+d_2}\right)+b_1}{c_1\left(\frac{a_2x+b_2}{c_2x+d_2}\right)+d_1}=\frac{(a_1a_2+b_1c_2)x+(a_1b_2+b_1d_2)}{(c_1a_2+d_1c_2)x+(c_1b_2+d_1d_2)}\]

In order to make $\Gamma$ symmetric, we also need to compute $\widetilde{f}$ for each constraint $x\le f(y)$ in $\Gamma$. The inverse of:

\[f(x)=\frac{ax+b}{cx+d}\]

\noindent is:

\[f^{-1}(x)=\frac{dx-b}{-cx+d}\]

\noindent so computing inverses doesn't increase the bit complexity of a fractional linear function.

So $H(\Gamma)$ has constraints that are piecewise fractional-linear with coefficients having at most $3\beta$ bits each. So $H^{\lceil\log_2(n)+1\rceil}(\Gamma)$ has constraints that are piecewise of the same form, where the coefficients are integers with at most $\mathcal{O}(n^{\log_2(3)}\beta)$ bits. Since the constraints should be continuous, the break points between two pieces are coincidence points of the functions on those pieces.

By \Cref{lem:minvalues}, the value of $\text{min}_{v_1\rightarrow x}(\text{max}(\text{range}(x)))$ is either $f_{v_1\rightarrow x}$, an attracting point of $f_{v_1\rightarrow v_1}$, or $f_{v_1\rightarrow z}(d)$ for some literal $z\ne v_1$ and $d$ an attracting point of $f_{z\rightarrow z}$. The functions $f_{x\rightarrow y}$ are just the functions that appear in $H^{\lceil\log_2(n)+1\rceil}(\Gamma)$, so are piecewise fractional-linear by the above discussion. 

The equation:

\[\frac{ax+b}{cx+d}=x\]

\noindent yields a quadratic equation:

\[cx^2+(d-a)x-b=0\]

\noindent after clearing denominators. So a fractional-linear function is either the identity function $f(x)=x$ or has at most $2$ fixed points. In particular, each attracting point of a piecewise fractional-linear function is an isolated fixed point of one of the pieces. So the attracting points of each $f_{x\rightarrow y}$ are of each form $p+q\sqrt{r}$, where $p$, $q$, and $r$ are integers with a number of bits at most polynomial in $n$ and $\beta$. 

Applying a factional linear map to a number of form $p+q\sqrt{r}$ yields a number of the same form, indeed:

\[\frac{a(p+q\sqrt{r})+b}{c(p+q\sqrt{r})+d}=\frac{((ap+b)+aq\sqrt{r})(cp+d-cq\sqrt{r})}{(cp+d)^2+rq^2c^2}=\frac{(ap+b)(cp+d)-acq^2r}{(cp+d)^2+rq^2c^2}+\frac{(ad-bc)}{(cp+d)^2+rq^2c^2}q\sqrt{r}\]

Since the endpoints of the interval $\text{range}(x)$ are rational, we conclude that $\text{min}_{v_1\rightarrow x}(\text{max}(\text{range}(x)))$ is a number of form $p+q\sqrt{r}$ where $p$, $q$, and $r$ have at most polynomially many bits. The same is true of locally extremal fixed points of $\text{min}_{v_1\rightarrow v_2}$ and points where $\text{min}_{v_1\rightarrow -v_1}(-c)=c$ or $\text{min}_{-v_1\rightarrow v_1}(c)=-c$ (note that $\text{min}_{v_1\rightarrow -v_1}$ is increasing, so there is at most one point where $\text{min}_{v_1\rightarrow -v_1}(-c)=c$). In particular, we can choose a value of $v_1$ that has this form where $v_1$ extends to a satisfying assignment of $\Gamma$.

We now repeat this process to set values of $v_2, \dots, v_n$. The difference now is that $\text{range}(v_1)$ may no longer be bounded by rational numbers, since we might have set $v_1$ to an irrational value. But this is fine, because applying $\text{min}_{v_2\rightarrow x}$ to a number of form $p+q\sqrt{r}$ still yields a number of the same form. Inductively, we conclude that if $\Gamma$ has a satisfying assignment, then it has an assignment where are variables have form $p+q\sqrt{r}$, where $p$, $q$, and $r$ are rational numbers with at most $\mathcal{O}\left(\beta n^{2\log_2(3)}\right)$ bits. 

Given an assignment of variables of this form, checking a constraint in $\Gamma$ requires at most a constant number of arithmetic operations with numbers that roots of polynomials of degree at most $2$. These arithmetic operations can be performed in polynomial time (see e.g. Mishra and Pedersen \cite{IrrationalArithmetic}), so a solution can be checked in polynomial time, proving that $\mathcal{M}-2\text{SAT}$ is in NP.
\end{proof}

By \Cref{thm:BBreduction}, this also proves that the boundary-boundary art-gallery problem is in NP. 

The boundary-boundary art-gallery problem is NP-hard, but the reduction in \Cref{thm:BBreduction} is nondeterministic, so it doesn't imply anything about the complexity of $\mathcal{M}-2\text{SAT}$. We show that there is a quasi-polynomial-time algorithm for $\mathcal{M}-2\text{SAT}$, so it is not NP-hard unless the exponential-time hypothesis fails.

\begin{theorem}\label{thm:quasipoly}
There is a quasi-polynomial-time algorithm for $\mathcal{M}-2\text{SAT}$
\end{theorem}

\begin{proof}
Let $\Gamma$ be an instance of $\mathcal{M}-2\text{SAT}$ with $n$ variables where each constraint is a piecewise fractional-linear function where the coefficients have at most $\beta$ bits each. For each pair of literals $x$ and $y$, we can take the minimum over all constraints of form $x\le f(y)$ in $\Gamma$, so assume that $\Gamma$ has at most one constraint $x\le f(y)$ for each such pair of literals. Let $k$ be the maximum number of pieces in a constraint in $\Gamma$.

As shown in the proof of \Cref{thm:npmembership}, the functions $f_{x\rightarrow y}$ are piecewise fractional linear where the coefficients of each piece are integers with at most $\mathcal{O}(n^{\log_2(3)}\beta)$ bits each. We also need a bound on the number of pieces. If $f$ and $g$ are monotone fractional-linear functions with at most $k$ pieces each, then the composition $f\circ g$ has at most $2k$ pieces. 

If $f_1, \dots, f_j$ are piecewise-fractional-linear functions with at most $k$ pieces each, then $\text{min}_i(f_i)$ has at most $(2j^3+j)k$ pieces. To see this, split $\rr$ at the endpoint of every piece in $f_i$. This splits $\rr$ into at most $j(k-1)+1$ chunks such that each $f_i$ restricts to a single piece on each chunk. Two fractional linear functions intersect in at most two places, so the number of intersection points on each of these chunks is at most $2j^2$.

Each application of $H$ involves computing minimums of sets of at most $2n$ functions, so each constraint in $H^{\lceil\log_2(n)+1\rceil}(\Gamma)$ has at most $n^{\mathcal{O}(\log(n))}k$ pieces, and computing all the functions $f_{x\rightarrow y}$ takes time $n^{\mathcal{O}(\log(n))}k\beta$. 

In order to decide if $\Gamma$ is decidable, we need to compute the functions $\text{min}_{x\rightarrow y}$, which correspond to tight paths in $V(\Gamma)$. Each loop in $V(\Gamma)$ is the infinite composition of some $f_{x\rightarrow x}$. The number of values taken by $f_{x\rightarrow x}^{\infty}$ is at most twice the number of pieces in $f$.

Let $G$ be a directed graph that has a vertex $(x, c)$ whenever $x$ is a literal in $\Gamma$ and $c$ is a value that can be taken by $f_{x\rightarrow x}^{\infty}$. $G$ has an edge from $(x, c)$ to $(y, d)$ whenever $f_{x\rightarrow x}^{\infty}(f_{x\rightarrow y}(d))=c$. For every pair of a literal $x$ and a vertex $(y, c)$ of $G$, define:

\[G(x\rightarrow (y, d))=\text{min}\{c : (x, c)\text{ is a vertex of }G \text{ and there is a path from }(x, c)\text{ to }(y, d)\text{ in }G\}\]

We now claim that:

\[\text{min}_{x\rightarrow y}(c)=\text{min}\{\text{id}, f_{x\rightarrow x}^{\infty}\}\circ \text{min}\left(\{f_{x\rightarrow y}(c)\}\cup\{f_{x\rightarrow z}(G(z\rightarrow (w, (f_{w\rightarrow w}^{\infty}\circ f_{w\rightarrow y})(c))) : z\text{ a literal}\}\right)\]

To see this, let $p$ be a tight path from $x$ to $y$ in $V(\Gamma)$. For each $c\in \rr$, we want to show that $p(c)$ is one of $f_{x\rightarrow y}(c)$, $(f_{x\rightarrow x}^{\infty}\circ f_{x\rightarrow y})(c)$, $f_{x\rightarrow z}(d)$, or $(f_{x\rightarrow x}^{\infty}\circ f_{x\rightarrow z})(d)$, for some literal $z$ and $d\in \rr$ such that there is a literal $w$ and path from $(z, d)$ to $(w, (f_{w\rightarrow w}^{\infty}\circ f_{w\rightarrow y})(c))$ in $G$. This is by induction on the length of $p$. 

If $p$ has length $1$, then it must be exactly $f_{x\rightarrow y}$. If $p$ is longer, then there are two cases based on whether the first edge in $p$ is a loop.

First, suppose we can write $p=f_{x\rightarrow x}^{\infty}\circ q$ for $q$ a tight path from $x$ to $y$. Since $q$ is shorter than $p$ and $f_{x\rightarrow x}\circ f_{x\rightarrow x}=f_{x\rightarrow x}$, we have that $p$ is has one of the required forms.

Otherwise, we can write $p=f_{x\rightarrow z}\circ f_{z\rightarrow z}\circ q$ for $z\ne x$ and $q$ a tight path from $z$ to $y$. By induction, $q$ has one of the types above. Since $f_{z\rightarrow z}^{\infty}\circ f_{z\rightarrow z}^{\infty}=f_{z\rightarrow z}^{\infty}$, we have that $p(c)$ is either $(f_{x\rightarrow z}\circ f_{z\rightarrow z}^{\infty}\circ f_{x\rightarrow y})(c)$ or $(f_{x\rightarrow z}\circ f_{z\rightarrow z}^{\infty}\circ f_{z\rightarrow z'})(d)$ where there is a path from $(z', d)$ to $(w, (f_{w\rightarrow w}^{\infty}\circ f_{w\rightarrow y})(c))$ in $G$.

There is a path from $(z, (f_{z\rightarrow z}^{\infty}\circ f_{x\rightarrow y})(c))$ to itself in $G$, so the first case can be written $f_{x\rightarrow z}(d)$ with $d$, $z$ satisfying the appropriate conditions. There is an edge from $(z, (f_{z\rightarrow z}^{\infty}\circ f_{z\rightarrow z'})(d))$ to $(z', d)$ in $G$, so there is a path from $(z, (f_{z\rightarrow z}^{\infty}\circ f_{z\rightarrow z'})(d))$ to $(w, (f_{w\rightarrow w}^{\infty}\circ f_{w\rightarrow y})(c))$ in $G$. So the second case also has the appropriate form.

For each of $f_{x\rightarrow y}(c)$, $(f_{x\rightarrow x}^{\infty}\circ f_{x\rightarrow y})(c)$, $f_{x\rightarrow z}(d)$, and $(f_{x\rightarrow x}\circ f_{x\rightarrow z})(d)$ for $z, d$ as above, there is a tight path in $V(\Gamma)$ that evaluates to that value at $c$. This proves the claim.

We can compute $G$ in time $n^{\mathcal{O}(\log(n))}k\beta$, and so we can calculate each $G(x\rightarrow (y, d))$ in quasi-polynomial time. The above calculation shows that this lets us calculate the functions $\text{min}_{x\rightarrow y}$. In order to determine if $\Gamma$ is satisfiable, we just need to test if any of the conditions in \Cref{lem:KeyLemma} occur. Since each $\text{min}_{x\rightarrow y}$ is a piecewise-fractional-linear (though not necessarily continuous) function, this is straightforward. The entire computation runs in time $n^{\mathcal{O}(\log(n))}k\beta$.
\end{proof}

The minimum of $k$ piecewise-fractional-linear functions with $k$ pieces each is the lower envelope of an arrangement of $nk$ arcs, where each pair of arcs has at most $2$ intersection points. Results on Davenport-Schinzel sequences (see Sharir and Agarwal \cite{DSBook}) imply the number of pieces in the minimum is at most $\mathcal{O}(jk2^{\alpha(jk)})$, where $\alpha$ is the inverse Ackermann function. Unless $k$ is much much larger than $n$, this improves the bound used in the proof of \Cref{thm:quasipoly}, but the result of the theorem does not change qualitatively. 

\section{Irrational Coordinates}\label{sec:irrationalexample}

In this section, we give a sketch of a construction of an instance of the boundary-boundary art-gallery problem where the unique solution with $3$ guards requires placing the guards at irrational coordinates. The instance is illustrated in \Cref{fig:IrrationalExample}.

\begin{figure}
\centering
\includegraphics[page=1,width=\textwidth]{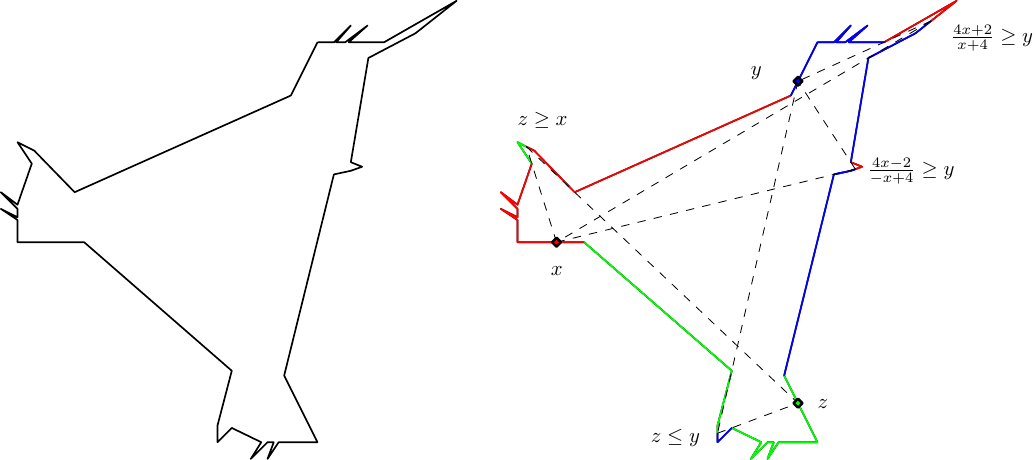}
\caption{Left: an art gallery whose boundary can be guarded with three guards on the boundary, but only if irrational coordinates are allowed. Right: guarding this example with $3$ guards. There are $4$ nooks that represent constraints between pairs of guards.}
\label{fig:IrrationalExample}
\end{figure}

The instance can be guarded with $3$ guards. The positions of the guards represent variables $x, y$ and $z$, and the $4$ nooks create constraints $z\ge x$, $y\ge z$, $\frac{4x+2}{x+4}\ge y$ and $\frac{4x-2}{-x+4}\ge y$. The constraints $z\ge x$ and $y\ge z$ imply that $y\ge x$, so the only two solutions are $x=y=z=\pm \sqrt2$, as illustrated by \Cref{fig:ConstraintPlot}. Only the solution $x=y=z=\sqrt{2}$ satisfies the range constraints created by the polygon.

\begin{figure}
\centering
\includegraphics[page=2]{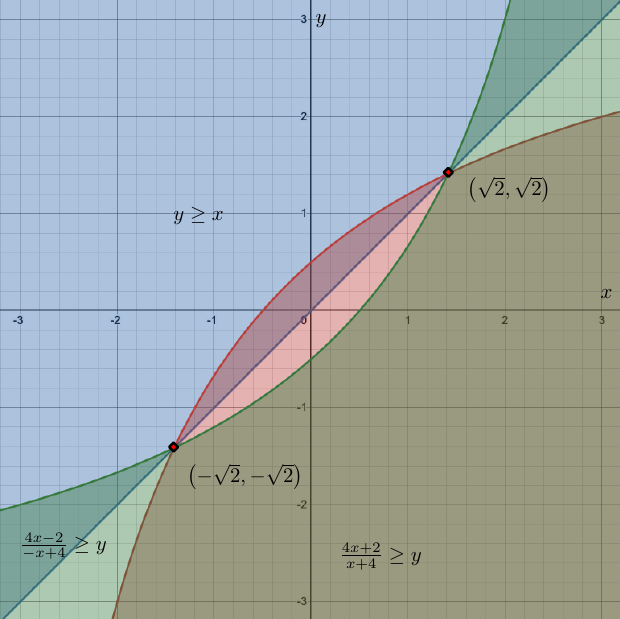}
\caption{The three regions created by constraints $y\ge x$, $\frac{4x+2}{x+4}\ge y$ and $\frac{4x-2}{-x+4}\ge y$. The only two points in all three regions are $\pm \left(\sqrt{2}, \sqrt{2}\right)$}
\label{fig:ConstraintPlot}
\end{figure}

We now give a brief sketch of the techniques used to construct this polygon. The key ideas are adapted from the construction in \cite{Stade2025}. There are $6$ small triangular slits in the polygon boundary, and the only way to guard all of the slits with $3$ guards total is to place one guard somewhere on each of the red segments shown in \Cref{fig:IrrationalStructure}. There are $4$ nooks (highlighted in blue in \Cref{fig:IrrationalStructure}) that each create a continuous constraint between two of the guard positions. By the discussion in the proof of \Cref{thm:BBreduction}, the constraint created by each nook is fractional-linear.

\begin{figure}
\centering
\includegraphics[page=3]{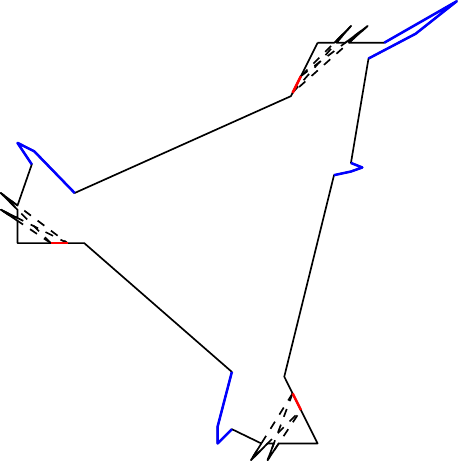}
\caption{The $6$ slits on the boundary of the polygon can only be guarded with $3$ guards if there is one guard on each of the red segments shown. The nooks that create continuous constraints are highlighted in blue.}
\label{fig:IrrationalStructure}
\end{figure}

Two fractional linear functions are identical if they agree on three different points (for $f$ and $g$ fractional linear, the identity $f(x)=g(x)$ becomes quadratic after clearing denominators). So in order to verify that a nook enforces the appropriate constraint, it is sufficient to check $3$ points. We illustrate this by a constructing a nook that enforces the constraint $\frac{4x+2}{x+4}\ge y$.

First, we choose two lines, so that points on those lines represent values of variables $x$ and $y$, as illustrated in \Cref{fig:NookConstruction1}. We choose label the points corresponding to $x=0$, $x=1$, and $x=2$. These should be mapped to $y=\frac12$, $y=\frac65$, and $y=\frac53$ by the nook. We draw a line through the points representing $x=0$ and $y=\frac12$, and choose a point $P$ on this line, as in \Cref{fig:NookConstruction1}.

\begin{figure}
\centering
\includegraphics[page=4]{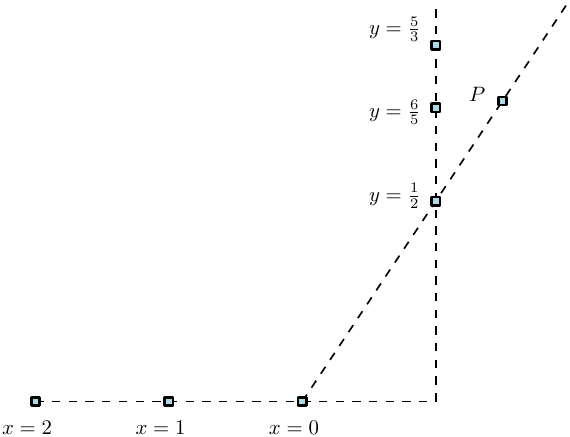}
\caption{A point on the vertical line represents a value of $y$ and a point on the horizontal line represents a value of $x$.}
\label{fig:NookConstruction1}
\end{figure}

Next, we choose a line $\ell$ through the point representing $x=0$, as illustrated in \Cref{fig:NookConstruction2}. We draw lines from the points representing $y=\frac65$ and $y=\frac53$ through $P$, and let points $I$ and $J$ respectively be the intersection points of these new lines with $\ell$. 

\begin{figure}
\centering
\includegraphics[page=5]{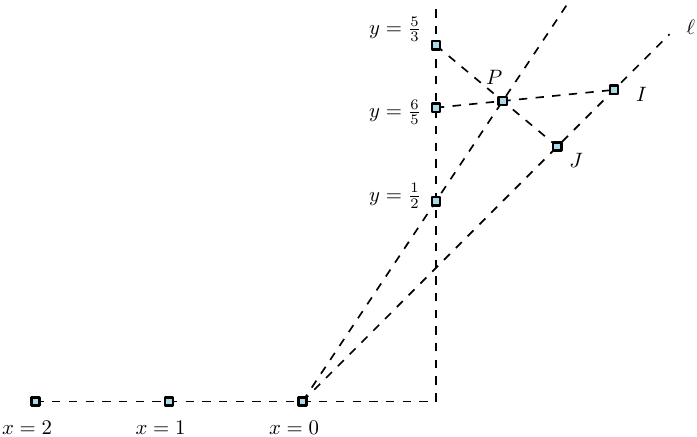}
\caption{The first step in the construction of a nook.}
\label{fig:NookConstruction2}
\end{figure}

Next, we draw a line from $I$ to $x=1$ and from $J$ to $x=2$, and let $Q$ be the intersection point of these lines, as in \Cref{fig:NookConstruction3}. We create a nook from the points $P$, $I$, $J$ and $Q$. Let $g$ be the function from the line representing $x$ to $\ell$ given by projecting through $Q$ and let $f$ be the function from $\ell$ to the line representing $y$ given by projecting through $P$. If two guards are placed at points representing values of $x$ and $y$ with $x\in [1, 2]$ and $y\in \left[\frac65, \frac53\right]$, then they together can see all of the nook if and only if $f(g(x))\le y$, as illustrated in \Cref{fig:NookConstruction4}.

\begin{figure}
\centering
\includegraphics[page=6]{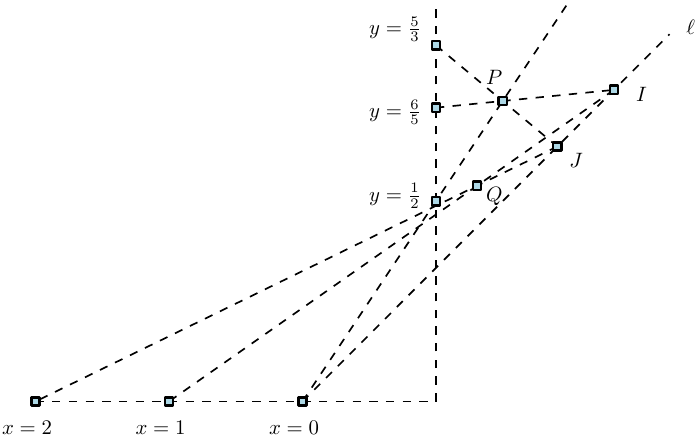}
\caption{The second step in the construction of a nook.}
\label{fig:NookConstruction3}
\end{figure}

\begin{figure}
\centering
\includegraphics[page=7]{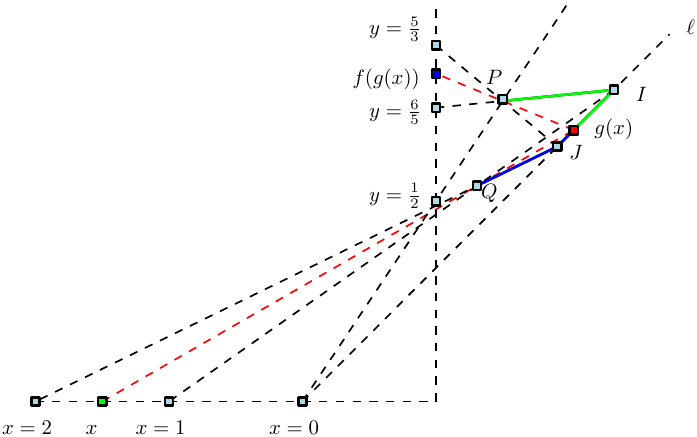}
\caption{In order to see the parts of the nook that cannot be seen by a guard at position $x$, a guard on the vertical line should be placed at a position representing a $y$ value of at most $f(g(x))$. The composition $f\circ g$ is a fractional linear function that maps $1$ to $\frac65$ and $2$ to $\frac53$ by construction. To see that it maps $0$ to $\frac12$, we notice that the line through $x=0$ and $Q$ intersects $\ell$ at the point $x=0$, and the line through $x=0$ and $P$ intersects the vertical line at $y=\frac12$ by construction. }
\label{fig:NookConstruction4}
\end{figure}

The composition $f\circ g$ is fractional linear, so in order to verify that $f(g(x))=\frac{4x+2}{x+4}$, we just need to check that $f(g(0))=\frac12$, $f(g(1))=\frac65$ and $f(g(2))=\frac53$. This can be easily seen from the figures. 

Each of the $4$ nooks in \Cref{fig:IrrationalExample} is constructed using a similar method. A detailed specification of the construction is shown in \Cref{fig:NookConstructionDetails}.

\begin{figure}
\centering
\includegraphics[page=8,width=\textwidth]{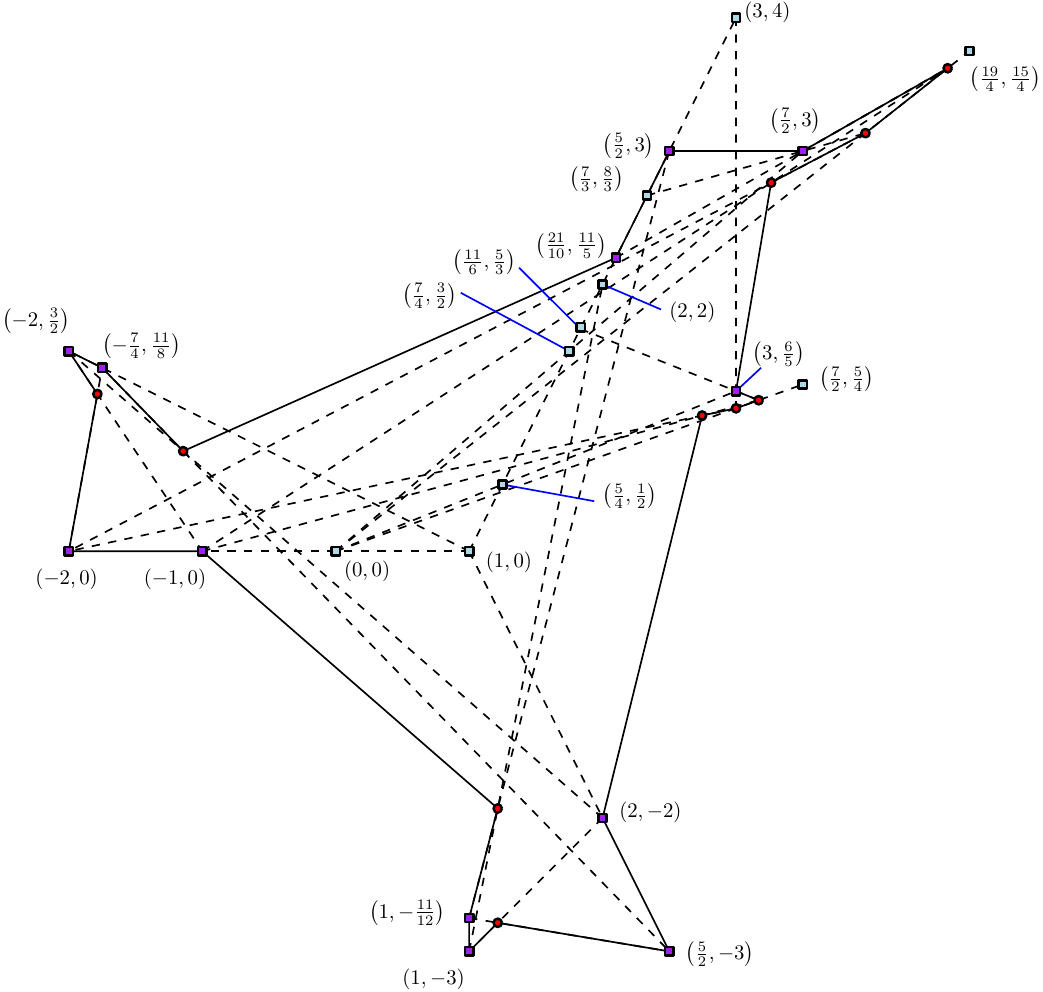}
\caption{Specification of the construction. The $6$ slits are omitted, since the exact parameters of the slits are not important. Each of the purple points is a vertex of the polygon with explicitly-given coordinates. Each of the blue points is an explicitly-given point that is used to help construct the polygon (but not a vertex of the polygon). The red points are derived vertices, which are defined as intersections of lines through some of the other points. A guard represents a value of $x$ if it is at point $(-x, 0)$, a value of $y$ if it is at $\left(\frac32+\frac12y, 1+y\right)$, or a value of $z$ if it is at point $\left(\frac32+\frac12z, -1-z\right)$.}
\label{fig:NookConstructionDetails}
\end{figure}

\section{Acknowledgments}

I am very grateful to Lucas Meijer and Tillman Miltzow for discussions that helped to produce the example in \Cref{sec:irrationalexample}, and to the anonymous reviewers for helpful feedback and comments. 

This work is supported by Starting Grant 1054-00032B from the Independent Research Fund Denmark under the Sapere Aude research career programme, and by the Carlsberg Foundation, grant CF24-1929.
The author is part of BARC, supported by VILLUM Foundation grant 16582.

\section{Conclusion}

We have shown that the boundary-boundary art-gallery problem is in NP, despite the fact that irrational coordinates are sometimes needed. Each of the X-Y art-gallery variants is now known to be either NP-complete or $\exists\rr$-complete. 

We have described new techniques for approaching 2CCSPs, in particular giving a quasi-polynomial-time algorithm for $\mathcal{M}-2\text{SAT}$ and showing that it is in NP. It is interesting to wonder if $\mathcal{M}-2\text{SAT}$ can be solved in polynomial time.

We could also ask if more complicated 2CCSPs might also be in NP. The techniques from \Cref{sec:ctsinference} can be applied to more complicated problems, but the numerical calculations required quickly become intractable in the Turing model. 

Even if we allow a model of computation that can perform more complicated arithmetic operations, it is unclear whether general 2CCSPs become easy. Our algorithms for $\mathcal{M}-2\text{SAT}$ exploit the fact that the class of fractional-linear functions is closed under composition. This seems to be in contrast to higher-degree polynomials, where the degree grows exponentially with the number of compositions. However, only a small fraction of all high-degree polynomials can be written as a composition of lower-degree polynomials, and it isn't obvious that the exponential degree increase leads to an exponential increase in combinatorial complexity. For example, it seems to be unknown whether the composition $f\circ g$ can have more than $\mathcal{O}(n)$ fixed points, when $f$ and $g$ are monotone polynomials of degree at most $n$.

\bibliographystyle{plainurl}
\bibliography{bib}

\end{document}